\documentclass[11pt,a4paper]{article}
\usepackage{amssymb}
\usepackage{amsmath}
\usepackage{amsfonts,bm}
\usepackage[pdftex]{graphicx}
\graphicspath{{./graphics/}}
\DeclareGraphicsExtensions{.pdf,.png,.jpg}
\usepackage{amscd}
\usepackage{a4wide}
\usepackage{microtype}
\usepackage{bbm}
\usepackage{slashed}
\usepackage{fancyhdr}
\usepackage{amsthm}
\usepackage{mathrsfs}
\usepackage{hyperref}
\usepackage{color} 
\usepackage{pstricks}
\usepackage{enumerate}
\usepackage{epstopdf}
\usepackage{calligra}
\usepackage{dsfont}

\definecolor{darkred}{rgb}{0.8,0.1,0.1}
\hypersetup{
     colorlinks=true,         
     linkcolor=darkred,
     citecolor=blue,
}

\usepackage{comment}
\usepackage[all,cmtip]{xy}

\usepackage{marvosym}

\usepackage{geometry}
\newtheorem{theorem}{\rmfamily\bfseries{Theorem}}[section]

\newtheorem{lemma}[theorem]{\rmfamily\bfseries{Lemma}} 
\newtheorem{proposition}[theorem]{\rmfamily\bfseries{Proposition}}

\theoremstyle{remark}

\newtheorem{remark}{Remark}

\newtheorem*{notation}{Notation}

\numberwithin{equation}{section}

\def\supp{\mathrm{supp}}

\def\dim{\mathrm{dim}}

\def\1{\mathbbm{1}}

\newcommand{\oA}{{\mathds A}}
\newcommand{\oB}{{\mathds B}}
\newcommand{\oC}{{\mathds C}}

\newcommand{\oN}{{\mathds N}}

\newcommand{\oR}{{\mathds R}}

\newcommand{\un}{{\mathds 1}}

\usepackage{xparse}

\NewDocumentCommand{\oldnorm}{sO{}m}{%
  {\IfBooleanTF{#1}
    {\oldnormaux{\left|}{\right|}{#3}}
    {\oldnormaux{#2|}{#2|}{#3}}}
}
\makeatletter
\newcommand{\oldnormaux}[3]{\mathpalette\oldnormaux@i{{#1}{#2}{#3}}}
\newcommand{\oldnormaux@i}[2]{\oldnormaux@ii#1#2}
\newcommand{\oldnormaux@ii}[4]{%
  \sbox\z@{$\m@th#1#2#4#3$}%
  \sbox\tw@{$\m@th\|$}%
  \mathopen{\hbox to\wd\tw@{\hss\vrule height \ht\z@ depth \dp\z@ width .3\wd\tw@\hss}}%
  #4
  \mathclose{\hbox to\wd\tw@{\hss\vrule height \ht\z@ depth \dp\z@ width .3\wd\tw@\hss}}%
}
\makeatother

\title{\bf{Spectral properties of the 2D magnetic Weyl-Dirac operator
with a short-range potential}}

\author{Magno B. Alves,$^{*}$ Oswaldo M. Del Cima,$^{\dag}$ \\ Daniel H.T. Franco$^{\dag}$
       and Emmanuel A. Pereira$^{\ddag}$ \vspace{4mm}\\
       $^{*}$ Departamento de Matem\'atica,\\
       Universidade Federal de Juiz de Fora,\\
      Juiz de Fora, MG, Brasil, CEP: 36.036-900.\vspace{4mm}\\
      $^\dag$ Departamento de F\'\i sica-CCE,\\
      Universidade Federal de Vi\c cosa,\\
      Vi\c cosa, MG, Brasil, CEP: 36570-900.\vspace{4mm}\\
      $^\ddag$ Departamento de F\'\i sica-ICEx\\
      Universidade Federal de Minas Gerais,\\
      Belo Horizonte, MG, Brasil, C.P. 702, CEP: 30161-970.\vspace{4mm}\\
{\small e-mail: \texttt{magno\underline{~}branco@yahoo.com.br, oswaldo.delcima@ufv.br,}} \\
{\small \texttt{daniel.franco@ufv.br, emmanuel@fisica.ufmg.br}}\vspace{4mm}\\
}

\date{\today}

\begin{document}

\maketitle

\begin{abstract}
This paper is devoted to the study of the spectral propperties of the Weyl-Dirac
or massless Dirac operators, describing the behavior of quantum quasi-particles
in dimension 2 in a homogeneous magnetic field, $B^{\rm ext}$, perturbed by a
chiral-magnetic field, $b^{\rm ind}$, with decay at infinity and a short-range scalar electric
potential, $V$, of the Bessel-Macdonald type. These operators emerge from the action of a pristine
graphene-like QED$_3$ model recently proposed in~\cite{WOE}. First, we establish
the existence of states in the discrete spectrum of the Weyl-Dirac operators between the zeroth and
the first (degenerate) Landau level assuming that $V=0$. In sequence, with $V_s \not= 0$, where $V_s$
is an attractive potential associated with the $s$-wave, which emerges when analyzing the $s$- and
$p$-wave M{\o}ller scattering potentials among the charge carriers in the pristine graphene-like QED$_3$
model, we provide lower bounds for the sum of the negative eigenvalues of the operators
$|\boldsymbol{\sigma} \cdot \boldsymbol{p}_{\boldsymbol{A}_\pm}|+ V_s$. Here, $\boldsymbol{\sigma}$
is the vector of Pauli matrices, $\boldsymbol{p}_{\boldsymbol{A}_\pm}=\boldsymbol{p}-\boldsymbol{A}_\pm$, 
with $\boldsymbol{p}=-i\boldsymbol{\nabla}$ the two-dimensional momentum operator and
$\boldsymbol{A}_\pm$ certain magnetic vector potentials. As a by-product of this, we have the stability
of bipolarons in graphene in the presence of magnetic fields.
\end{abstract}

\vspace{3mm}



\section{Introduction}
\label{Sec1}
\hspace*{\parindent}
The pristine graphene, a monolayer of pure graphene, is a gapless quasi-bidimensional system behaving like 
a half-filling semimetal where the quasi-particles, charge carriers, can be described by a two-dimensional massless
Dirac operator, with the speed of light $c$ being replaced by the Fermi velocity, $v_F \approx 10^{-2}c$.
Due to its unusual properties it has attracted a great deal of attention since its discovery. Such exciting properties
and perspectives are a direct consequence of the fact that the low-energy properties of quasi-particles in graphene
can be described by the model based on the continuum limit of the tight binding approximation which obeys a
relation formally identical to the massless Dirac equation in (1+2)-dimensions, with the holes and the pseudospin
states of the $\oA$ and $\oB$ sublattices being the counterparts of the positrons and the spin,
respectively. For this reason, this genuinely two-dimensional material provides a bridge between condensed matter
physics and quantum electrodynamics in (1+2)-dimensions. 

In the current paper, we focus on the planar massless Dirac operators with a Bessel-Macdonald potential
that emerge from the action of a pristine graphene-like QED$_3$ model proposed in Refs.~\cite{{WOE,Em}},
associated with two kinds of massless fermions $\Psi_+$ and $\Psi_-$ -- each of them describing electron-polaron
(electron-phonon) and hole-polaron (hole-phonon) quasi-particles -- where the subscripts $+$ (sublattice $\oA$)
and $-$ (sublattice $\oB$) are related to the two inequivalent $\boldsymbol{K}$ and $\boldsymbol{K'}$
points in the Brillouin zone of a monolayer graphene. We will analyze these operators in the presence of magnetic 
fields; namely, the energy operators corresponding to a quasi-particle in a force field of another quasi-particle and
subject to relativistic effects in the presence of magnetic fields are of the form (for simplicity the units are chosen
so that $\hbar=v_F=1$)
\begin{align}
H_\pm(x)=
-i\boldsymbol{\sigma} \cdot (\boldsymbol{p}-\boldsymbol{A}_\pm)
+V\,\,,
\label{OpBR}
\end{align}
with $V(x)=\gamma K_0(\beta |x|)$ being a short-range potential of the Bessel-Macdonald type. The operators
$H_\pm$ act on the two components of the spinors $\Psi_\pm$ and we shall say that $\psi_\pm$ are the
{\em upper components} and $\vartheta_\pm$ the {\em lower components}, respectively. In (\ref{OpBR}),
$\boldsymbol{\sigma}=(\sigma_1;\sigma_2)$ are the Pauli $2 \times 2$-matrices
\begin{align*}
\sigma_1=
\begin{pmatrix}
0 & ~~1 \\[3mm]
1 & ~~0
\end{pmatrix}
\quad , \quad
\sigma_2=
\begin{pmatrix}
0 & -i \\[3mm]
i & ~~0
\end{pmatrix}\,\,,
\end{align*}
and $\boldsymbol{D}_{\boldsymbol{A}_\pm}=\boldsymbol{\nabla}-i\boldsymbol{A}_\pm$ are the magnetic
gradients. $\boldsymbol{A}_\pm \overset{\rm def.}{=} e\boldsymbol{A}^{\rm ext} \pm g\boldsymbol{a}^{\rm ind}$,
where $\boldsymbol{A}^{\rm ext}$ and $\boldsymbol{a}^{\rm ind}$ are, respectively, vector potentials associated with a high
external and homogeneous magnetic field $B^{\rm ext}$ and a perturbation {\em induced} within the bulk of the system
$b^{\rm ind}$~\cite{WOE,Em}, called chiral-magnetic field~\cite{Ind-b}). The constants $e$ and $g$ are, in this order,
the coupling constants associated with the electric and chiral charges~\cite{WOE,Em}. Throughout this paper, we
assume that the vector potentials $\boldsymbol{A}^{\rm ext},\boldsymbol{a}^{\rm ind} \in L_p^{\rm loc}(\oR^2;\oR^2)$,
for some $2 \leqslant p \leqslant \infty$, and satisfy the well-known relations $B^{\rm ext}={\rm curl}~\boldsymbol{A}^{\rm ext}$
and $b^{\rm ind}={\rm curl}~\boldsymbol{a}^{\rm ind}$, which are understood in the sense of distributions.\footnote{Of course,
this is necessary if $\boldsymbol{A}^{\rm ext}$ and $\boldsymbol{a}^{\rm ind}$ are not differentiable.} Hereafter, each statement
containing the double subscript ``$\pm$'' must be understood separately for the upper subscript and the lower one.
This will allow one to state the results simultaneously for both operators $H_+$ and $H_-$.

\begin{remark}
We shall assume that the operators $H_\pm$ admit self-adjoint realizations, which are still denoted by
$H_\pm$ in $L_2(\oR^2;\oC^2)$.
\end{remark}

\begin{remark}
Unlike the case in $\oR^3$, in $\oR^2$ the fields $B^{\rm ext}$ and $b^{\rm ind}$ are pseudoscalars.
For the constant magnetic field $B^{\rm ext}$ orthogonal to the $x_1,x_2$-plane, we will fix the sign
$B^{\rm ext} > 0$ without loss of generality, since we can change the sign of $B^{\rm ext}$ by changing the
coordinates $(x_1,x_2) \to (x_2,x_1)$. Thus, with $B^{\rm ext} > 0$ it follows that the vector potential
$\boldsymbol{A}^{\rm ext}$ is given by $\boldsymbol{A}^{\rm ext}=(A_1^{\rm ext},A_2^{\rm ext})=\frac{B^{\rm ext}}{2}(-x_2,x_1)$.
In turn, in order to avoid an abrupt discontinuity of the chiral-magnetic field $b^{\rm ind}$ at the boundary of the system, we
will assume that $b^{\rm ind}$ has the profile described in Figure \ref{Induced-b-field}. Note that the support of the chiral-magnetic
field $b^{\rm ind}$ is the closed disk defined by $\supp~b^{\rm ind}=\overline{D} \overset{\rm def}{=}\overline{D(0;R_2)}
=\bigl\{x \in \oR^2 \mid |x| \leqslant R_2\bigr\}$. In other words, we are assuming that
$b^{\rm ind}$ is smooth, compactly supported and {\em decays sufficiently fast} outside the system. For example,
if we define $R_2=R_1+\varepsilon$, then we can write
\begin{align}
b^{\rm ind}(|x|)=
\begin{cases}
b^{\rm ind}\,\,, \quad & \text{if $R_1 \leqslant |x|$}\,\,; \\[3mm]
b^{\rm ind} \cdot {\exp}\Biggl[-\frac{\varepsilon}{R_2-|x|}
{\exp}\left(-\frac{\varepsilon}{|x|-R_1}\right)\Biggr]
\,\,,\quad & \text{if $R_1 < |x| < R_2$}\,\,; \\[3mm]
0\,\,, \quad & \text{if $|x| \geqslant R_2$}\,\,.
\end{cases}
\label{Induced-b-fieldA}
\end{align}

\begin{figure}
\begin{center}
\includegraphics[scale=0.50]{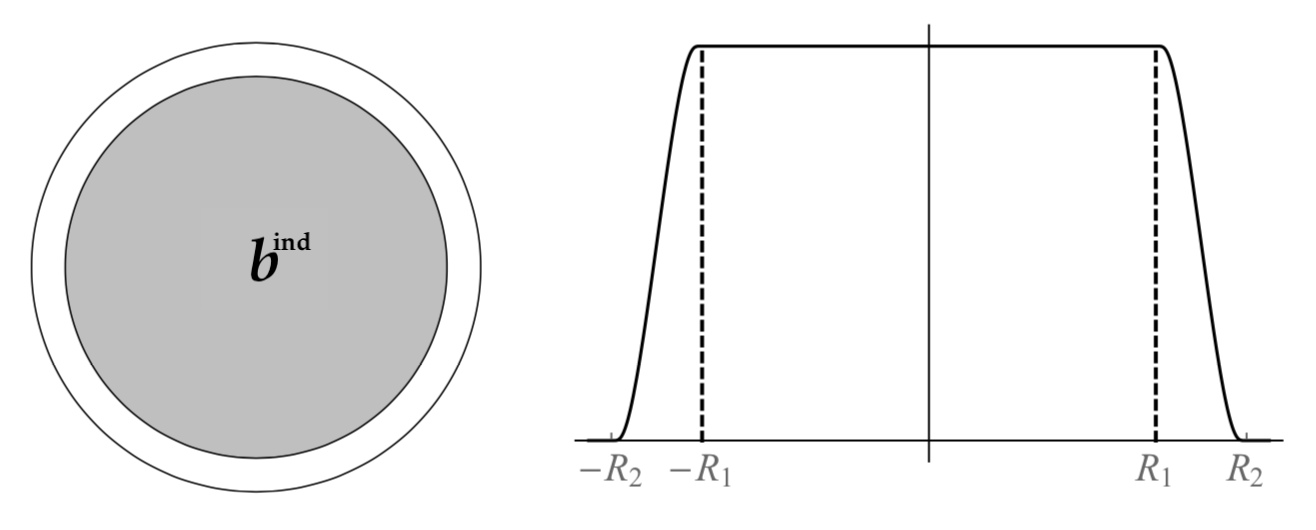}
\caption{Small sample of graphene in the shape of a disk. The support of chiral-magnetic field
$b^{\rm ind}$ (left) and a radial cross section (right).}
\label{Induced-b-field}
\end{center}
\end{figure}
\end{remark}

In (\ref{OpBR}), $\gamma > 0$ is the coupling parameter taken to be contained in the non-negative semi-axis $[0,\infty)$
and $K_0(\beta |x|)$ is the Bessel-Macdonald potential induced by
\begin{align*}
K_0(\beta |x|)
=\frac{1}{2} \int _{0}^{\infty} e^{-{\frac {\pi \beta^2 |x|^{2}}{\eta}}}
e^{-\frac{\eta}{4\pi}}~\frac{1}{\eta}\,d\eta\,\,,
\end{align*}
where $\beta > 0$ is a real parameter, which has an inverse length dimension (whose precise meaning is
given in Remark \ref{RemaK0} below). The $K_0$-potential arises from the parity-preserving
$U_A(1) \times U_a(1)$ massless QED$_3$ proposed in~\cite{{WOE,Em}} when analyzing the $s$- and $p$-wave
M{\o}ller scattering potentials among the charge carriers, written as $V_s(x)=\frac{1}{2\pi} (e^2-g^2)K_0(\beta |x|)$ and 
$V_p(x)=\frac{1}{2\pi} (e^2+g^2)K_0(\beta |x|)$, respectively. While the $p$-wave state fermion-fermion (or antifermion-antifermion)
scattering potential shows to be {\bf repulsive} whatever the values of the electric $(e)$ and chiral $(g)$ charges, for $s$-wave
scattering of fermion-fermion (or antifermion-antifermion), the interaction potential might be {\bf attractive} provided
$|g| > |e|$ -- here we take the magnitude of electric and chiral charges because these charges can take on positive
and negative values depending on the spin value of the spinors $\Psi_\pm$ as displayed in Table \ref{table1}.

\begin{table}
\begin{center}
\begin{tabular}{||c||c||c||c||c||}
\hline
Spinor & Electric charge & Chiral charge & Spin & Quasi-particle \\
\hline\hline
$\Psi_+$ & $-e$ & $-g$ & $+1/2$ & electron-polaron \\
\hline
$\Psi_+$ & $+e$ & $+g$ & $-1/2$ & hole-polaron \\
\hline
$\Psi_-$ & $-e$ & $+g$ & $-1/2$ & electron-polaron \\
\hline
$\Psi_-$ & $+e$ & $-g$ & $+1/2$ & hole-polaron \\
\hline
\end{tabular}
\end{center}
\caption[]{Electric and chiral charges of the quasi-particles (see more details in Ref.~\cite{WOE})).}
\label{table1}
\end{table}

The question of whether or not the attractive $s$-wave state potential favours $s$-wave massless bipolarons (two-fermion
bound states) has been answered in Ref.~\cite{MOD1}, where for a suitably projected two-dimensional massless Dirac
operator in the presence of a Bessel-Macdonald potential {\em without a magnetic field}, it has been proved the absence
of bound states if $\gamma \leqslant \gamma_{\rm crit}$ (the subcritical region where the matter is stable).

\begin{remark}
At this point, it is important to note that the typical length-scale of the interaction between the charge carriers in
graphene in the conduction band is orders of magnitude in nanometers~\cite{Kim}. This result indicates that, {\em necessarily},
the interaction between the charge carriers in graphene must be described by a short-range potential. Since graphene is a
{\em strictly} two-dimensional material~\cite{Novo}, this implies that the interaction among the massless fermion quasi-particles
is nonconfining, so the vector meson mediated quasi-particles contained in the model in~\cite{WOE,Em}, namely the photon and
the N\'eel quasi-particles, must be massive. In this case, the parameter $\beta$ that multiplies the argument of the $K_0$-function
has inverse length dimension, thus fixing a length scale, an interaction range, which is related to the mass of the boson-mediated
quantum exchanged during the two quasi-particle scattering (see~\cite{WOE,Em} and references therein). It is also worth noting
that massless mediated quanta in three space-time dimensions yield logarithm-type (confining) interaction potentials~\cite{Maris}.
Moreover, the Coulomb potential $1/|x|$ adopted in literature is a long-range potential, which contradicts, as mentioned above,
the experimental fact~\cite{Kim} that the interaction between the quasi-particles in graphene is short-range.
\label{RemaK0}
\end{remark}

Regarding the chiral magnetic field $b^{\rm ind}$ and the attractive potential $V_s(x)=-\frac{1}{2\pi} (g^2-e^2) K_0(\beta |x|)$
(associated with the $s$-wave) we would like to point out that they belong to the class of electromagnetic perturbations
considered in Ref.~\cite{Stock}, namely:

\,\,\,$(A1)$ $B_\pm=eB^{\rm ext} \pm gb^{\rm ind}$, where $B^{\rm ext} > 0$ and $b^{\rm ind} \in L_p^{\rm loc}(\oR^2;\oR)$,
for some $1 \leqslant p \leqslant \infty$, and $\displaystyle{\lim_{\kappa \to \infty} \|\chi_{_{\{|x| \geqslant \kappa\}}}} 
b^{\rm ind}\|_\infty=0$;

\,\,\,$(A2)$ $V_s \in L_p^{\rm loc}(\oR^2;\oR)$ for some $2 \leqslant p \leqslant \infty$ and 
$\displaystyle{\lim_{\kappa \to \infty} \|\chi_{_{\{|x| \geqslant \kappa\}}}} V\|_\infty=0$.

\,\,\,Here $\chi_{_{\{|x| \geqslant \kappa\}}}$ denotes the characteristic function on the set
$\{|x| \geqslant \kappa\}$. Assuming that $B_\pm$ fulfills $(A1)$ we can always find
$\boldsymbol{A}_\pm \in L_q^{\rm loc}(\oR^2;\oR^2)$, for some $2 \leqslant q \leqslant \infty$,
satisfying $B_\pm={\rm curl}~\boldsymbol{A}_\pm$.

\,\,\,The aim of this article is to establish some spectral properties of the operators $H_\pm$.
It is organized as follows. In Section \ref{Sec3}, in particular, the location of the essential spectrum, consisting
of isolated eigenvalues with infinite multiplicity (called quantum Landau levels), is obtained. Assuming that $V=0$,
the existence of states in the discrete spectrum of $H_\pm$ between the (degenerate) Landau levels is analysed
in Section \ref{Sec4}. The results of Sections  \ref{Sec3} and \ref{Sec4} are similar to those found by
K\"onenberg-Stockmeyer~\cite{Stock}. Our proofs follow the ideas developed in~\cite{Stock}. However, it must be
pointed out that the results of~\cite{Stock} hold for a 2D magnetic massless Dirac operator in the presence of a
perturbed homogeneous magnetic field without taking into account the fermionic spin degree of freedom of the
quasi-particles. From this point of view, our results can be seen as an extension of~\cite{Stock}. Section \ref{Sec5}
is dedicated to the proof of the magnetic Lieb-Thirring type inequality, {\em i.e.}, we provide lower bounds for the sum of
the negative eigenvalues of the operators $|\boldsymbol{\sigma} \cdot \boldsymbol{p}_{\boldsymbol{A}_\pm}|+V_s$.
Here, $\boldsymbol{p}_{\boldsymbol{A}_\pm}=\boldsymbol{p}-\boldsymbol{A}_\pm$, with
$\boldsymbol{p}=-i\boldsymbol{\nabla}$ the two-dimensional momentum operator, and
$V_s(x)=-\frac{1}{2\pi} (g^2-e^2) K_0(\beta |x|)$. For that, we benefit from the insights and arguments found
in~\cite{GoFo,Sen}. As a by-product of this, we have the stability of bipolarons in magnetic fields.

\begin{notation}
The notation $\|\Psi\|_p$ stands for the norm of $\Psi$ in the space $L_p(\oR^2)$, $p \geqslant 1$.
If $p=2$, we usually write simply $\|\Psi\|$ and for $p=\infty$, one uses $\|\Psi\|_\infty=\sup_{x \in \oR^2} |\Psi(x)|$. 
Notation $\langle\,\cdot\,,\,\cdot\,\rangle$ stands for the inner product in $L_2(\oR^2)$. For a self-adjoint operator $T$,
we shall use the notation $\boldsymbol{\mathfrak{Dom}}(T)$ for the domain of $T$. Finally, we will use $C$ to denote
constants, which are not necessarily the same at each occurrence, which may depend on $p,n$, {\em etc}.
\end{notation}

\section{Essential spectrum and compact perturbations}
\label{Sec3}
\hspace*{\parindent}
In this Section, the location of the essential spectrum of the operators (\ref{OpBR}), consisting of isolated eigenvalues
with infinite multiplicity, is obtained. The rest of the spectrum (the discrete spectrum) will be analyzed in Sections \ref{Sec4}
and \ref{Sec5}. We start by remembering that in quantum mechanics, Landau quantization refers to the quantization of the
cyclotron orbits of charged particles in a uniform magnetic field. As a result, the charged particles can only occupy orbits
with discrete, equidistant energy values, {\em i.e.}, the spectrum consists of eigenvalues of infinite multiplicity, the so-called
{\em Landau levels}, lying at the points of an arithmetic progression. These levels are degenerate, with the number of charged
particles per level directly proportional to the strength of the applied magnetic field.

Under a weak perturbation of the constant magnetic field, the eigenvalues, except the lowest one, may split, producing
a discrete spectrum between the Landau levels and a cluster of these eigenvalues around the Landau levels~\cite{Roze1,Roze2}.
In the case of the operators (\ref{OpBR}), assuming $V=0$, this splitting was found, with the energy spectrum of
operators $-i \boldsymbol{\sigma} \cdot \boldsymbol{D}_{\boldsymbol{A}_\pm}$ consisting of
degenerated eigenvalues that take the form~\cite{WOE,Em},
\begin{align}\label{LandauSpec}
\Lambda_{n,+,s}&= \pm\sqrt{2(eB^{\rm ext}+gb^{\rm ind})\left(n-s+1/2\right)}\,\,, \nonumber \\[3mm]
\Lambda_{n,-,s}&= \pm\sqrt{2(eB^{\rm ext}-gb^{\rm ind})\left(n-s+1/2\right)}\,\,,
\end{align}
with $n \in \oN$, where $\Lambda_{n,+,s}$ and $\Lambda_{n,-,s}$ are the Landau levels associated to $\Psi_+$
(at sublattice $\oA$) and $\Psi_-$ (at sublattice $\oB$) for electron-polarons (if $+$ sign in
$\Lambda_{n,+,s}$ and $\Lambda_{n,-,s}$) or hole-polarons (if $-$ sign in $\Lambda_{n,+,s}$ and $\Lambda_{n,-,s}$)
and $s=\pm\frac{1}{2}$ are the so-called sublattice spin (the pseudospin eigenvalues related to sublattices) and not
the real electron spin ({\em cf.} Ref.~\cite{MR}). The conceptual novelty is that the presence of $b^{\rm ind}$,
seen as a perturbation, leads to the splitting of eigenvalues $\Lambda_{n,+,s}$ and $\Lambda_{n,-,s}$, which mimic the four-fold
broken degeneracy effect of the Landau levels (Figure \ref{Landau_4_fold}) experimentally observed in pristine graphene under
high applied magnetic fields; see {\em e.g.}~\cite{WOE,Em} and references therein. In the case of graphene, each of these
levels is four times degenerate because of the spin and the sublattice degeneracy. 

\begin{figure}
\begin{center}
\includegraphics[scale=0.09]{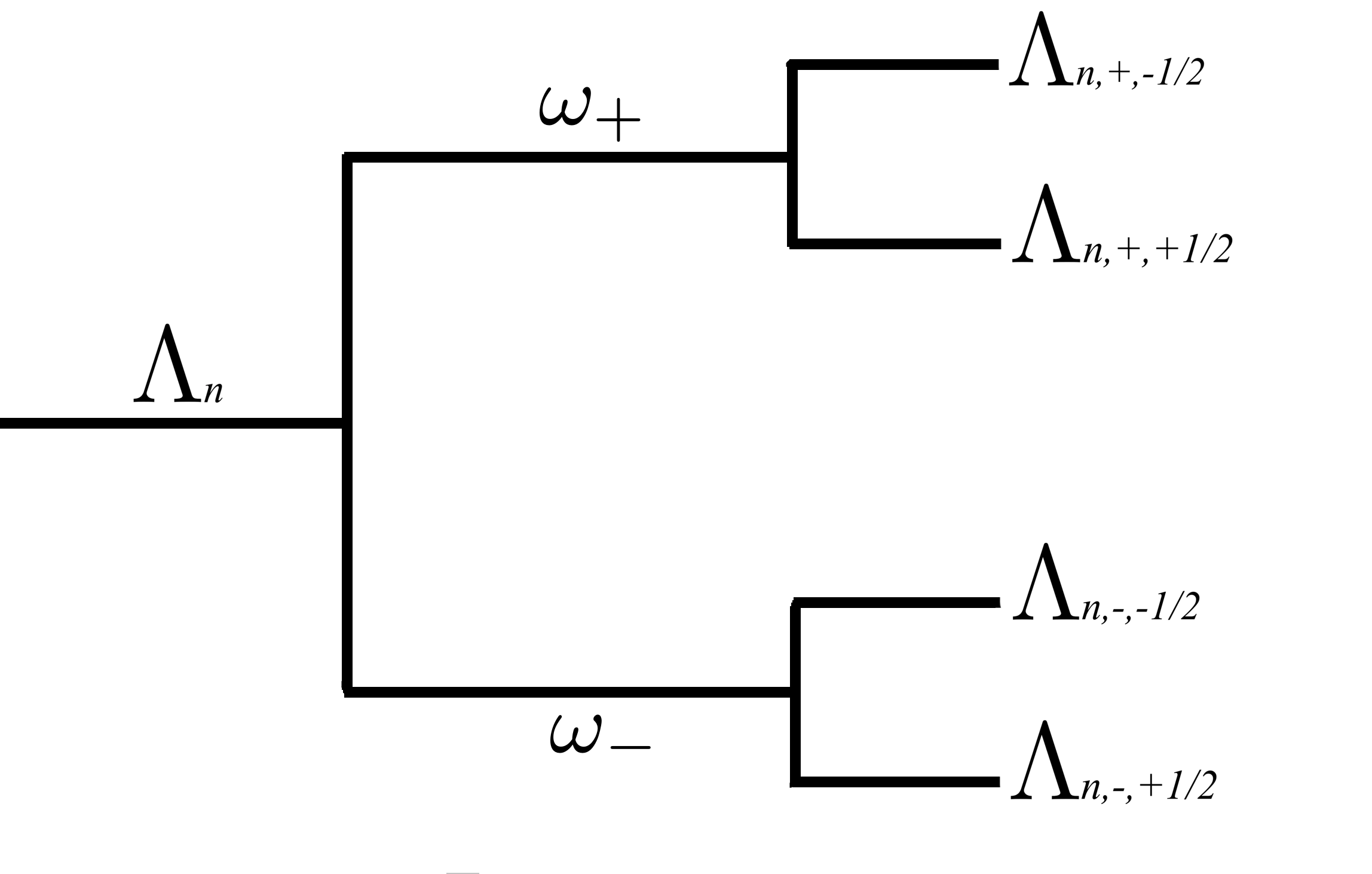}
\caption{Four-fold Landau levels of electron-polarons and hole-polarons at sublattices $\boldsymbol{A}$ and $\boldsymbol{B}$ 
with cyclotron frequencies $\omega_+=\sqrt{2(eB^{\rm ext}+gb^{\rm ind})}$ and $\omega_-=\sqrt{2(eB^{\rm ext}-gb^{\rm ind})}$,
respectively, provided $n \geqslant1$.}
\label{Landau_4_fold}
\end{center}
\end{figure}

It is important to emphasize that the lowest Landau level $(n = 0)$ appears at $\Lambda_{0,+,s}=\Lambda_{0,-,s}=0$ and
accommodates electron-polarons or hole-polarons with only one pseudospin eigenvalue, namely $s=+1/2$, signalizing a
possible anomalous-type quantum Hall effect (QHE) (the discovery of the anomalous QHE is the most direct evidence
of massless fermions in graphene~\cite{KN}). All other levels $n \geqslant 1$ are occupied by electron-polarons or
hole-polarons with both $(s=\pm 1/2)$ pseudospin eigenvalues. Therefore, this implies that for the lowest Landau level
$n=0$ the degeneracy is half of that for any other $n \geqslant 1$, likewise, all Landau levels $(n \geqslant 1)$ have the
same degeneracy (a number of electron-polaron or hole-polaron states with a given energy) but the zero-energy $(n=0)$
Landau level is shared equally by electron-polarons and hole-polarons, that is, depending on the sign of the applied
magnetic field there is only sublattice $\oA$ or sublattice $\oB$ states which contribute to the zero-energy
(lowest) Landau level.

\begin{remark}
In Ref.~\cite{Stock} although the authors consider a massless two-dimensional Dirac operator in the presence of a perturbed
homogeneous magnetic field $B=B_0+b$, the splitting in the four-fold Landau levels does not show up in the spectrum analysis
since they are considering only spinless quasi-particles. In turn, in Ref.\cite{KMSZ}, despite the authors considering two spinors, each
one related to the two inequivalent $\boldsymbol{K}$ and $\boldsymbol{K'}$ points in the Brillouin zone, the splitting in the
four-fold Landau levels also does not appear due to the fact that they are considering only one unperturbed magnetic field,
associated with a unique symmetry $U_A(1)$.
\end{remark}

\begin{remark}
In the absence of an external magnetic field, Schmidt~\cite{Sch} proved  that the essential spectrum of
massless Dirac operators with a rotationally symmetric potential (such as the $K_0$-potential) in two dimensions covers
the whole real line.
\end{remark}

Our approach to obtaining the essential spectrum of the operators (\ref{OpBR}) is mainly based on the study by
Rozenblum-Tashchiyan~\cite{Roze1,Roze2} and K\"onenberg-Stockmeyer~\cite{Stock}. In this way, we reproduce
the reasoning from~\cite{Roze1,Roze2,Stock}. First of all, it is useful to introduce the complex variable $z=x_1+ix_2$
and to define the ``creation'' and ``annihilation'' operators, respectively,
\begin{align}\label{CAOper}
D_\pm^* \overset{\rm def.}{=} -2i \partial_z+i \frac{B_\pm}{2} \overline{z}
\quad \text{and} \quad
D_\pm \overset{\rm def.}{=} -2i \partial_{\overline{z}}-i \frac{B_\pm}{2} z\,\,,
\end{align}
where  $\partial_z \overset{\rm def.}{=} \frac{1}{2}(\partial_{x_1}-i\partial_{x_2})$ and 
$\partial_{\overline{z}} \overset{\rm def.}{=} \frac{1}{2}(\partial_{x_1}+i\partial_{x_2})$. 

The operators (\ref{CAOper}) can also be expressed by means of the scalar potential of the magnetic field,
the function $\Phi$, solving the equation $\Delta \Phi_\pm=B_\pm$:
\begin{align*}
D_\pm^*=-2i e^{\Phi_\pm} \partial_z\,e^{-\Phi_\pm}
\quad \text{and} \quad
D_\pm=-2i e^{-\Phi_\pm} \partial_{\overline{z}}\,e^{\Phi_\pm}\,\,.
\end{align*}
Here, $\Phi_\pm=\Phi^\circ \pm \varphi$, where $\Phi^\circ(\overline{z},z)=\frac{eB^{\rm ext}}{4}\overline{z}z$.
On the other hand, since we are assuming that $b^{\rm ind}$ is smooth and compactly supported, then it can
be shown that the scalar potential $\varphi$ for the field $b^{\rm ind}$ must be solution of the equations 
$4\partial_z \varphi=b^{\rm ind} \overline{z}$ and $4\partial_{\overline{z}} \varphi=b^{\rm ind} z$.

It can be easily found that the creation and annihilation operators $D_\pm^*,D_\pm$ satisfy the following relation
\begin{align}
[D_\pm,D_\pm^*]=2B_\pm=2(eB^{\rm ext} \pm gb^{\rm ind})\,\,.
\label{CAOper1}
\end{align}

With the help of operators (\ref{CAOper}), the operators $-i \boldsymbol{\sigma} \cdot \boldsymbol{D}_{\boldsymbol{A}_\pm}$
take the very simple forms
 \begin{align*}
-i \boldsymbol{\sigma} \cdot \boldsymbol{D}_{\boldsymbol{A}_\pm}
\overset{\rm def.}{=}D_{B_\pm}
=\begin{pmatrix}
0 & D_\pm^* \\[3mm]
D_\pm & 0
\end{pmatrix}
\quad \bigl(\text{on $\boldsymbol{\mathfrak{Dom}}(D_{B_\pm})
=\boldsymbol{\mathfrak{Dom}}(D_\pm) \oplus \boldsymbol{\mathfrak{Dom}}(D_\pm^*)$}\bigr)\,\,.
\end{align*}
According to Thaller~\cite[Theorem 5.13]{Thaller} (see also~\cite{Thaller1,Stock}), the operators $D_{B_\pm}$ can be
diagonalized by a suitable unitary Foldy-Wouthuysen transformation, $U_{\rm FW}$, defined by
\begin{align*}
U_{\rm FW}=\frac{1}{\sqrt{2}}\Bigl[{\un}_{2 \times 2}
+\sigma_3~{\rm sgn}\bigl(D_{B_\pm}\bigr)\Bigr]
\quad \text{and} \quad
U_{\rm FW}^*=\frac{1}{\sqrt{2}}\Bigl[{\un}_{2 \times 2}
-\sigma_3~{\rm sgn}\bigl(D_{B_\pm}\bigr)\Bigr]\,\,,
\end{align*}
where ${\rm sgn}\bigl(D_{B_\pm}\bigr)=D_{B_\pm}/|D_{B_\pm}|$ on
$\boldsymbol{\mathfrak{Ker}}(D_{B_\pm})^\perp$ and equals zero on 
$\boldsymbol{\mathfrak{Ker}}(D_{B_\pm})$, with $|D_{B_\pm}|
=\bigl(D_{B_\pm}^2\bigr)^{1/2}$, and
\begin{align*}
\sigma_3=
\begin{pmatrix}
1 & ~~0 \\[3mm]
0 &-1
\end{pmatrix}\,\,.
\end{align*}
A direct computation yields
\begin{align}\label{D_FW}
\bigl(D_{B_\pm}\bigr)_{\rm FW}=
U_{\rm FW}
\bigl(D_{B_\pm}\bigr)
U_{\rm FW}^*=
\begin{pmatrix}
\sqrt{D_\pm^* D_\pm} & 0 \\[3mm]
0 & -\sqrt{D_\pm D_\pm^*}
\end{pmatrix}\,\,.
\end{align}

The operator ${\rm sgn}\bigl(D_{B_\pm}\bigr)$ is a unitary map from $\boldsymbol{\mathfrak{Ker}}(D_{B_\pm})^\perp$
onto $\boldsymbol{\mathfrak{Ker}}(D_{B_\pm})^\perp$. In the standard representation the operator
${\rm sgn}\bigl(D_{B_\pm}\bigr)$ is given by (see~\cite[p.144]{Thaller})
\begin{align*}
{\rm sgn}\bigl(D_{B_\pm}\bigr)=
\begin{pmatrix}
0 & S^* \\[3mm]
S & 0
\end{pmatrix}
\quad \bigl(\text{on $\boldsymbol{\mathfrak{Ker}}(D_{B_\pm})^\perp$}\bigr)\,\,.
\end{align*}
Hence, it follows immediately from the trivial calculation
\begin{align}\label{D_FW1}
\begin{pmatrix}
D_\pm^* D_\pm & 0 \\[3mm]
0 & D_\pm D_\pm^*
\end{pmatrix}
=D_{B_\pm}^2
={\rm sgn}\bigl(D_{B_\pm}\bigr) D_{B_\pm}^2 {\rm sgn}\bigl(D_{B_\pm}\bigr)
=\begin{pmatrix}
S^*D_\pm D_\pm^* S & 0 \\[3mm]
0 & S D_\pm^* D_\pm S^*
\end{pmatrix}
\end{align}
which holds on $\boldsymbol{\mathfrak{Ker}}(D_{B_\pm})^\perp$ and which shows that $D_\pm D_\pm^*$ and 
$D_\pm^* D_\pm$ are mapped onto each other by the isometry $S$. 

An important property that follows from the relation (\ref{D_FW1}) is the coincidence of the spectra
$\sigma\bigl(D_\pm^* D_\pm\bigr)=\sigma\bigl(D_\pm D_\pm^*\bigr)$ except at zero. This, along with Eq.(\ref{D_FW}),
implies immediately, from the spectral mapping theorem, the following (see~\cite[Proposition 1]{Stock})

\begin{proposition}
Let $\boldsymbol{A}_\pm \in L_2^{\rm loc}(\oR^2;\oR^2)$ satisfying $B_\pm={\rm curl}~\boldsymbol{A}_\pm$.
Then, the spectrum of $D_{B_\pm}$ is symmetric with respect to zero and
\begin{align*}
\sigma_\#(-i\boldsymbol{\sigma} \cdot \boldsymbol{D}_{\boldsymbol{A}_\pm}) \cap (0,\infty)
=\sigma_\#\bigl(\sqrt{D_\pm^* D_\pm}\bigr) \setminus \{0\}\,\,,
\quad \text{with $\# \in \{{\rm disc},{\rm ess}\}$}\,\,.
\end{align*}
\label{Prop1}
\end{proposition}

In the case of the unperturbed magnetic massless Dirac operator, $i\boldsymbol{\sigma} \cdot \boldsymbol{D}_{\boldsymbol{A}^{\rm ext}}$,
we have $B_\pm \equiv eB^{\rm ext}$ and
\begin{align*}
d^* \overset{\rm def.}{=} -2i \partial_z+i \frac{eB^{\rm ext}}{2} \overline{z}
\quad \text{and} \quad
d \overset{\rm def.}{=} -2i \partial_z-i \frac{eB^{\rm ext}}{2} \overline{z}\,\,.
\end{align*}
It is known that $dd^*$ and $d^*d$ are self-adjoint with domains
$\boldsymbol{\mathfrak{Dom}}(dd^*)=\bigl\{\psi \in \boldsymbol{\mathfrak{Dom}}(d^*)
\mid d^*\psi \in \boldsymbol{\mathfrak{Dom}}(d)\bigr\}$ and 
$\boldsymbol{\mathfrak{Dom}}(d^*d)=\bigl\{\psi \in \boldsymbol{\mathfrak{Dom}}(d)
\mid d\psi \in \boldsymbol{\mathfrak{Dom}}(d^*)\bigr\}$. In addition, as above, there is a unitary map $S$
from $\boldsymbol{\mathfrak{Ker}}(dd^*)^\perp$ to $\boldsymbol{\mathfrak{Ker}}(d^*d)^\perp$, such that
$dd^*=S\bigl(d^*d)S^*$. 

The representation of $d^*$ and $d$ via the scalar potential takes the form
\begin{align*}
d^*=-2i e^{\Phi^\circ} \partial_z\,e^{-\Phi^\circ}
\quad \text{and} \quad
d=-2i e^{-\Phi^\circ} \partial_{\overline{z}}\,e^{\Phi^\circ}\,\,,
\end{align*}
where $\Phi^\circ(\overline{z},z)=\frac{eB^{\rm ext}}{4}\overline{z}z$. The equation $d^*d \psi=0$ is equivalent to 
$d \psi=0$, or $\partial_{\overline{z}}\,(e^{\Phi^\circ}\psi)=0$. So the function $\varphi=e^{\Phi^\circ}\psi$ is an entire analytic
function such that $\psi=e^{-\Phi^\circ}\varphi \in L_2$. The space of entire functions with this property is, obviously,
infinite-dimensional, it contains at least all polynomials in $z$. Proceeding in the standard way we can define
the following functions:
\begin{align*}
\psi^{(n)}=(d^*)^n \psi^{(0)}\,\,,
\end{align*}
where $\psi^{(0)}$ obeys $d\psi^{(0)}=0$. It is straightforward to check that $\psi^{(n)}$ obey the
eigenvalue equations
\begin{align*}
d^* d \psi^{(n)}=2neB^{\rm ext} \psi^{(n)}
\quad \text{and} \quad
dd^* \psi^{(n)}=2(n+1)eB^{\rm ext} \psi^{(n)}\,\,.
\end{align*}
In short, the operators $d^*,d$ act between Landau subspaces ${\mathscr L}_n=(d^*)^n {\mathscr L}_0$, $n \in \oN$,
\begin{align*}
d^*:{\mathscr L}_n \mapsto {\mathscr L}_{n+1}\,\,,
\quad d:{\mathscr L}_n \mapsto {\mathscr L}_{n-1}\,\,,
\quad d:{\mathscr L}_0 \mapsto 0\,\,,
\end{align*}
and are, up to constant factors, isometries of Landau subspaces.

Returning to perturbed magnetic massless Dirac operator, we note first that by~\cite[Lemma 1]{Stock} the operator of
multiplication by $|b^{\rm ind}|^{1/2}$ is relatively compact with respect to $\sqrt{\boldsymbol{p}^2+1}$,
$D_\pm D_\pm^*$ and $D_\pm^* D_\pm$. Therefore, as it follows from the relative compactness of the perturbation, by Weyl's
Theorem~\cite[Theorem XIII.14]{RS}, the essential spectrum of the operators
$-i\boldsymbol{\sigma} \cdot \boldsymbol{D}_{\boldsymbol{A}_\pm}$ is invariant under any
compact perturbation and consists of the same Landau levels of the unperturbed magnetic massless
Dirac operator $-i\boldsymbol{\sigma} \cdot \boldsymbol{D}_{\boldsymbol{A}^{\rm ext}}$, {\em i.e.}, the essential spectrum of 
$D_\pm D_\pm^*$ is just the one of $D_\pm^* D_\pm$, shifted by $2eB^{\rm ext}$. Moreover,
due to our previous discussion, we see that $0 \in \sigma_{\rm ess}(d^* d)$. That $0$ is an isolated point of
$\sigma(d^*d)$ follows by noting that, since $0 \notin \sigma_{\rm ess}(dd^*)$, $0$ is neither an accumulation
point of $\sigma(d^* d)$ nor of $0 \in \sigma(dd^*)$. In particular, all this together with Proposition \ref{Prop1}
leads us to the following

\begin{proposition}
Given that $B_\pm=eB^{\rm ext} \pm gb^{\rm ind}$ satisfies $(A1)$ and
$\boldsymbol{A}_\pm \in L_2^{\rm loc}(\oR^2;\oR^2)$, such that $B_\pm={\rm curl}~\boldsymbol{A}_\pm$,
then,
\begin{align*}
\sigma_{\rm ess}(-i\boldsymbol{\sigma} \cdot \boldsymbol{D}_{\boldsymbol{A}_\pm})
=\sigma_{\rm ess}(D_\pm^* D_\pm)
=\sigma_{\rm ess}(-i\boldsymbol{\sigma} \cdot \boldsymbol{D}_{\boldsymbol{A}^{\rm ext}})
=\sigma_{\rm ess}(d^* d)
=\sqrt{2eB^{\rm ext}\left(n-s+1/2\right)}\,\,,
\end{align*}
with $n \in \oN$ and $s=\pm\frac{1}{2}$. For each value of $n$, there are two states with that same energy, the state
with $n$ and $s=1/2$ and the state with $n+1$ and $s=-1/2$. Moreover, $0$ is an isolated point of
$\sigma(-i\boldsymbol{\sigma} \cdot \boldsymbol{D}_{\boldsymbol{A}_\pm})$ and $\sigma(d^*d)$.
\label{Prop2}
\end{proposition}

Now we add another perturbation by the $K_0$-potential. Arguing as before for $b^{\rm ind}$, since the operator of
multiplication by $V(x)=-\gamma K_0(\beta |x|)$ satisfies $(A2)$, then $V$ is relatively compact
with respect to $\sqrt{\boldsymbol{p}^2+1}$~\cite[Theorem 4.1]{MOD2} and the operators (\ref{OpBR}) have the same
essential spectra as the respective unperturbed ones. This immediately gives us as a consequence from~\cite[Lemma 1]{Stock}
the following

\begin{proposition}
Given that $V(x)=-\gamma K_0(\beta |x|)$ satisfies $(A2)$, then $V$ is relative
$(-i\boldsymbol{\sigma} \cdot \boldsymbol{D}_{\boldsymbol{A}_\pm})$-compact and
$\sigma_{\rm ess}(-i\boldsymbol{\sigma} \cdot \boldsymbol{D}_{\boldsymbol{A}_\pm})=
\sigma_{\rm ess}(-i\boldsymbol{\sigma} \cdot \boldsymbol{D}_{\boldsymbol{A}_\pm}+V)$.
\label{Prop3}
\end{proposition}

\section{Discrete spectrum of the purely magnetic operator}
\label{Sec4}
\hspace*{\parindent}
The next proposition specifies conditions on $gb^{\rm ind}$ under which the operators
$-i\boldsymbol{\sigma} \cdot \boldsymbol{D}_{\boldsymbol{A}_\pm}$ have states in the discrete spectrum.
As our arguments can be considered as an extension of those of K\"onenberg-Stockmeyer~\cite{Stock},
for the sake of completeness, we reproduce the proof of Lemma 3 from~\cite{Stock}, limiting ourselves to
pointing out the main difference between the two proofs.

\begin{proposition}
Define $\lambda_{n,s}=\sqrt{2eB^{\rm ext}\left(n-s+1/2\right)}$, with $n \in \oN$ and $s=\pm\frac{1}{2}$. Assume that
$B_\pm$ satisfies $(A1)$ and let $\boldsymbol{A}_\pm \in L_2^{\rm loc}(\oR^2;\oR^2)$ such that
$B_\pm={\rm curl}~\boldsymbol{A}_\pm$. Then, we have

\,\,\,$(a)$ If $gb^{\rm ind} < 0$ on some open set, then
\begin{align*}
\dim\Bigl(\boldsymbol{\mathfrak{Ran}}\bigl(\chi_{_{(\lambda_{0,s},\lambda_{1,s})}}(-i\boldsymbol{\sigma} \cdot \boldsymbol{D}_{\boldsymbol{A}_\pm})\bigr)\Bigr)
=\dim\Bigl(\boldsymbol{\mathfrak{Ran}}\bigl(\chi_{_{(-\lambda_{1,s},\lambda_{0,s})}}(-i\boldsymbol{\sigma} \cdot \boldsymbol{D}_{\boldsymbol{A}_\pm})\bigr)\Bigr)
=\infty\,\,.
\end{align*}

$(b)$ If $gb^{\rm ind} > 0$, then
\begin{align*}
\dim\Bigl(\boldsymbol{\mathfrak{Ran}}\bigl(\chi_{_{(\lambda_{0,s},\lambda_{1,s})}}(-i\boldsymbol{\sigma} \cdot \boldsymbol{D}_{\boldsymbol{A}_\pm})\bigr)\Bigr)
=\dim\Bigl(\boldsymbol{\mathfrak{Ran}}\bigl(\chi_{_{(-\lambda_{1,s},\lambda_{0,s})}}(-i\boldsymbol{\sigma} \cdot \boldsymbol{D}_{\boldsymbol{A}_\pm})\bigr)\Bigr)
=0\,\,.
\end{align*}
\label{Prop4}
\end{proposition}

\begin{proof}
Part $(a)$. Let $D(0;R_2)$ be an open disk with
$\supp~b^{\rm ind}=\overline{D(0;R_2)}=\bigl\{x \in \oR^2 \mid |x| \leqslant R_2\bigr\}$.
Recall that there are infinitely many functions $\omega$ analytic in $z$ (at least all polynomials in z), with
$\Psi_\pm=e^{-\Phi_\pm} \omega \in \boldsymbol{\mathfrak{Ker}}(D_\pm^*D_\pm)$. At this point lies the main
difference between the model studied in \cite{Stock} and the model studied in this article. In light of the model
proposed in Refs.~\cite{{WOE,Em}}, for such $\Psi_\pm$, as we are assuming that $b^{\rm ind}$ is a strictly
positive function (see Eq.(\ref{Induced-b-fieldA})), depending on the sign of the chiral charge $g$ (see Table \ref{table1}),
there will be only sublattice ${\oA}$ or sublattice ${\oB}$ states which will contribute to the zero-energy (lowest) Landau
level. Thus, if $gb^{\rm ind}$ is strictly negative, we have, using (\ref{CAOper1}),
\begin{align}\label{KoSt}
\langle \Psi_\pm,D_\pm D_\pm^* \Psi_\pm \rangle
&=2\langle \Psi_\pm,(eB^{\rm ext} \pm gb^{\rm ind}) \Psi_\pm \rangle \nonumber \\[3mm]
&\leqslant 2eB^{\rm ext} \|\Psi_\pm\|^2 + 2g \int_{\overline{D(0;R_2)}} b^{\rm ind}(|x|)
|\Psi_\pm(x)|^2\,\,dx \\[3mm]
&\leqslant 2eB^{\rm ext} \|\Psi_\pm\|^2\,\,. \nonumber
\end{align}
where in the last inequality we use the fact that $\Psi_\pm$ cannot vanish on $D(0;R_2)$. Let $(\Psi_\pm^{(n)})_{n \in \oN}$
be an orthonormal system such that $\Psi_\pm^{(n)} \in \boldsymbol{\mathfrak{Ker}}(D_\pm^*D_\pm)$, namely,
$\Psi_\pm^{(n)}=z^n e^{-\Phi_\pm}$. For $N \in \oN$ define the self-adjoint matrix
\begin{align*}
M_N=\left(\left\langle \Psi_\pm^{(n)},D_\pm D_\pm^* \Psi_\pm^{(m)} \right\rangle \right)_{1 \leqslant n,m \leqslant N}\,\,.
\end{align*}
It follows from (\ref{KoSt}) that $M_N < 2eB^{\rm ext}$. The Rayleigh-Ritz variational principle implies
\begin{align*}
0 \leqslant \lambda_n(D_\pm D_\pm^*) \leqslant \lambda_n(M_n) < 2eB^{\rm ext}\,\,,
\end{align*}
with $n=1,\ldots,N$. For some self-adjoint operator $T$, it is well-known that if $\lambda_1 \leqslant \lambda_2 \leqslant \lambda_3 \cdots$ 
are the eigenvalues of $T$ below the essential spectrum, respectively, the infimum of the essential spectrum,
once there are no more eigenvalues left, then
\begin{align*}
\lambda_n(T)=\sup_{\Psi_1,\ldots,\Psi_{n-1}} \inf_{\Psi_1 \in U(\Psi_1,\ldots,\Psi_{n-1})} \langle \Psi,T \Psi \rangle\,\,,
\end{align*}
where
\begin{align*}
U(\Psi_1,\ldots,\Psi_{n-1})=
\Bigl\{\Psi \in \boldsymbol{\mathfrak{Dom}}(T) \mid \|\Psi\|=1, \Psi \in {\rm span}\{\Psi_1,\ldots,\Psi_{n-1}\}^\perp\Bigr\}\,\,.
\end{align*}
Hence, since $N$ is arbitrary, the mini-max principle implies that
\begin{align*}
\dim\Bigl(\boldsymbol{\mathfrak{Ran}}\bigl(\chi_{_{(\lambda_{0,s},\lambda_{1,s})}}(D_\pm D_\pm^*)\bigr)\Bigr)
=\dim\Bigl(\boldsymbol{\mathfrak{Ran}}\bigl(\chi_{_{(-\lambda_{1,s},\lambda_{0,s})}}(D_\pm D_\pm^*)\bigr)\Bigr)
=\infty\,\,,
\end{align*}
for $0 \notin \sigma_{\rm ess}(D_\pm D_\pm^*)$ by Proposition \ref{Prop2}. The claim is now a consequence of
Proposition \ref{Prop1} and (\ref{D_FW1}).

Part $(b)$. If $gb^{\rm ind}$ is strictly positive, depending on the sign of the magnetic field $b^{\rm ind}$ and of the
chiral charge $g$, we have that $D_\pm D_\pm^* \geqslant 2eB^{\rm ext}$, since
$D_\pm D_\pm^*-D_\pm^* D_\pm=2B_\pm \geqslant 2eB^{\rm ext}$. Thus, again, the claim follows from
Proposition \ref{Prop1} and (\ref{D_FW1}).
\end{proof}

\section{Bounds for the sum of negative eigenvalues}
\label{Sec5}
\hspace*{\parindent}
Of course, the main interesting situation is to consider the massless Dirac operators perturbed by an electric potential
(the $K_0$-potential in our case), corresponding to the interaction between the charge carriers in the conduction band
and its evolution under the action of a magnetic field. As noted in the Introduction, the question of whether or not the
attractive $s$-wave state potential favours $s$-wave massless bipolarons (two-fermion bound states) has been answered
in Ref.~\cite{MOD1}, where for a suitably projected two-dimensional massless Dirac operator in the presence of a
Bessel-Macdonald potential without a magnetic field, it has been proved the absence of bound states if
$\gamma \leqslant \gamma_{\rm crit}$ (the subcritical region where the matter is stable). This is in agreement with the fact
that Weyl-Dirac fermions cannot immediately form bound states by electrostatic potentials. This section considers the
possibility that two charged quasi-particles with an attractive short-ranged potential between them which is not strong
enough to form bound states, may bind in presence of a magnetic field. In other words, we study the energy of quasi-particles
confined to a finite region in the graphene layer via a magnetic field and interacting via the
$K_0$-potential. The possible emergence of two-quasi-particle bound states draws attention to superconductivity in
graphene~\cite{CFFetal,CFDetal,ISTetal}. Thus, the physical applications of graphene and the spirit of universality
encompassed in the original Lieb-Thirring inequality lead us to search for magnetic Lieb-Thirring type inequality on the sum
of negative eigenvalues of the operators $|\boldsymbol{\sigma} \cdot \boldsymbol{p}_{\boldsymbol{A}_\pm}|+V_s$. Here,
$|\boldsymbol{\sigma} \cdot \boldsymbol{p}_{\boldsymbol{A}_\pm}|$ is the so-called {\em massless relativistic Pauli
operator}~\cite{GoFo}, $V_s(x)=-\frac{1}{2\pi} (g^2-e^2) K_0(\beta |x|)$ is the potential of the Bessel-Macdonald type
(associated with the $s$-wave) and $\boldsymbol{p}_{\boldsymbol{A}_\pm}=\boldsymbol{p}-{\boldsymbol{A}_\pm}$,
with $\boldsymbol{p}=-i\boldsymbol{\nabla}$ the two-dimensional moment operator. For brevity, from now on we
shall use the notation $\boldsymbol{P}_{\boldsymbol{A}_\pm}$ for the Dirac operator
$\boldsymbol{\sigma} \cdot \boldsymbol{p}_{\boldsymbol{A}_\pm}$.

\subsection{Magnetic Lieb-Thirring inequality in $\oR^2$}
\label{Sec5a}
\hspace*{\parindent}
We shall not attempt to give an overview of this vast and beautiful subject, since it is much better to refer the
reader to the book of Lieb-Seiringer~\cite{LiSe}. We start recalling that the article by Lieb {\em et al.}~\cite[Theorem 5.1]{LSY}
contains for the Pauli operator, that is, the non-relativistic operator describing the motion of a particle with spin in a
constant magnetic field $B \equiv {\rm const}$, the inequality in $\oR^2$
\begin{align}
\sum_k |\lambda_k((\boldsymbol{\sigma} \cdot (\boldsymbol{p}-\boldsymbol{A}))^2+V)| \leqslant
C_{1} \int_{\oR^2} V_-^2(x)\,\,dx
+C_{2} |B| \int_{\oR^2} V_-(x)\,\,dx\,\,, 
\label{Lieb}
\end{align}
where $\lambda_k$, $k \in \oN$, denotes the negative eigenvalues of the Pauli operator,
enumerated in the non-decreasing order counting multiplicity, while $V_-$ denotes the negative part of $V$.

\begin{remark}
At this point, we recall that any function $f$ can be written as
\begin{align*}
f(x)=f_+(x)-f_-(x)\,\,,
\end{align*}
where the positive part of $f$ is defined by the formula
\begin{align*}
f_+(x)=\max\{f(x),0\}=
\begin{cases}
f(x)\,\,, & \text{if} \quad f(x) > 0 \\[3mm]
0\,\,, & \text{otherwise}
\end{cases}\,\,,
\end{align*}
while the negative part of $f$ is defined by the formula
\begin{align*}
f_-(x)=\max\{-f(x),0\}=-\min\{f(x),0\}=
\begin{cases}
-f(x)\,\,, & \text{if} \quad f(x) < 0 \\[3mm]
0\,\,, & \text{otherwise}
\end{cases}\,\,.
\end{align*}
A peculiarity of terminology is that the ``negative part'' is not really negative. Indeed, with the above convention
$f_+$ and $f_-$ are non-negative functions, that is, $f_+ \geqslant 0$ and $f_- \geqslant 0$; in addition, we have
$|f|=f_+ + f_-$. Hence, $f_\pm=(|f| \pm f)/2$. Naturally, our convention the double subscript ``$\pm$'' must be
understood differently from the one adopted above.
\label{Rema}
\end{remark}

A natural generalization of (\ref{Lieb}) for non-homogeneous magnetic fields was suggested by Erd\H{o}s~\cite{Erd}
(see also Erd\H{o}s-Solovej~\cite{ES}), with $|B|$ replaced by a so-called ``effective'' (scalar) magnetic field ${\mathfrak b}(x)$.
The problem of the effective field observed by Erd\H{o}s was first succesfully addressed by Sobolev~\cite{Sob}, and later
by Bugliaro {\em et al.}~\cite{BFFGS} and Shen~\cite{Sen}. In particular, Sobolev~\cite[Theorem 2.4]{Sob} obtained the
following estimate in $\oR^2$:
\begin{align}
\sum_k |\lambda_k((\boldsymbol{\sigma} \cdot (\boldsymbol{p}-\boldsymbol{A}))^2+V)| \leqslant
C_{1} \int_{\oR^2} V_-^2(x)\,\,dx
+C_{2} \int_{\oR^2} {\mathfrak b}(x)\,V_-(x)\,\,dx\,\,,
\label{Sobolev}
\end{align}
where ${\mathfrak b}(x)$ is the effective magnetic field, with the constants $C_1$ and $C_2$ independents of
$V,B,{\mathfrak b}(x)$.

Regarding the massless relativistic Pauli operator $|\boldsymbol{P}_{\boldsymbol{A}_\pm}|$,
we want to be able to obtain a result concerning the magnetic Lieb-Thirring inequality on the sum of negative
eigenvalues of the operator $|\boldsymbol{P}_{\boldsymbol{A}_\pm}|+V_s$, where $V_s$ is the
potential of the Bessel-Macdonald type (associated with the $s$-wave). Indeed, we now state the main
result of Section \ref{Sec5}:

\begin{theorem}
Denote by $\lambda_1, \lambda_2,\ldots$ the negative eigenvalues $($if any$)$ of the operator
$|\boldsymbol{P}_{\boldsymbol{A}_\pm}|+V_s$ defined on $L_2(\oR^2;\oC^2)$, the space of wave functions
of a single $($pseudo$)$spin-$\frac{1}{2}$ quasi-particle. Given that $B_\pm=eB^{\rm ext} \pm gb^{\rm ind}$,
suppose that $|B^{\rm ext}| \in L_\infty(\oR^2)$ and $|b^{\rm ind}| \in L_p(\oR^2)$ for some $p > 4/3$. Then, for
$V_- \overset{\rm def.}{=} \gamma_s K_0(\beta|x|)$, where $\gamma_s=\frac{1}{2\pi}(g^2-e^2)$ is the coupling constant
for the $s$-wave state, there exist constants $C_1$ and $C_2$, independent of $V_s,B^{\rm ext},b^{\rm ind}$, such that
\begin{align}
\sum_k |\lambda_k^\pm(|\boldsymbol{P}_{\boldsymbol{A}_\pm}|+V_s)| \leqslant
C_{1} \gamma_s^3\,\beta^{-2}+C_{2} \gamma_s\,\beta^{-2}\Bigl(\|eB^{\rm ext}\|_\infty+\|gb^{\rm ind}\|_p^{3p/(3p-4)}\Bigr)\,\,.
\label{Sobolev1}
\end{align}
\label{Sobolev2}
\end{theorem}

\begin{remark}
For the power of moments of the negative eigenvalues equal to 1, Theorem \ref{Sobolev2} for the massless
relativistic Pauli operator can be seen as a 2D version similar to Theorem 2.1 for the non-relativistic
Pauli operator in~\cite{Erd}.
\end{remark}

This theorem is proven in the next  subsection. As a strategy to prove it, we shall benefit from the insights
and arguments found in~\cite{GoFo,Sen}. First, let us remember that
$\sum_k |\lambda_k^\pm(|\boldsymbol{P}_{\boldsymbol{A}_\pm}|+V_s)|={\rm Tr}(|\boldsymbol{P}_{\boldsymbol{A}_\pm}|+V_s)$,
then we use the

\begin{theorem}[Lieb-Siedentop-Solovej~\cite{LSS}, Theorem 3, Appendix A]
Let $p \geqslant 1$ and suppose that $A$ and $B$ are two non-negative, self-adjoint linear operators on a
separable Hilbert space such that $(A^{p}-B^{p})_-^{1/p}$ is trace class. Then $(A-B)$ is also trace class and
\begin{align*}
{\rm Tr}(A-B)_- \leqslant {\rm Tr}(A^{p}-B^{p})_-^{1/p}\,\,.
\end{align*}
\end{theorem}

Originally, this inequality is due to Birman-Koplienko-Solomyak~\cite{BKS}, and for this reason it is known as
{\em BKS inequalities}. Here, we will take into account the form that is relevant to us, namely (see~\cite{GoFo})
\begin{align}
{\rm Tr}(A-B)_- \leqslant {\rm Tr}(A^{1/\alpha}-B^{1/\alpha})_-^\alpha\,\,,
\label{BKSinq}
\end{align}
for any positive self-adjoint operators $A$ and $B$ and $0 < \alpha < 1$. As an illustration of the usefulness of the
trace estimate (\ref{BKSinq}), for $s=1/2$, let us to effectively replace $|\boldsymbol{P}_{\boldsymbol{A}_\pm}|$
by $\boldsymbol{P}_{\boldsymbol{A}_\pm}^2=\boldsymbol{p}_{\boldsymbol{A}_\pm}^2-B_\pm$.
This will allow us then adapt arguments found in~\cite{Sen} in order to use the non-relativistic Lieb-Thirring inequality
(for the magnetic momentum $\boldsymbol{p}_{\boldsymbol{A}_\pm}$) to obtain lower bounds for the sum of the negative
eigenvalues of the operator $|\boldsymbol{P}_{\boldsymbol{A}_\pm}|+V_s$.

To reach these lower bounds, in the second step, we will compare the massless Pauli operator
$\boldsymbol{P}_{\boldsymbol{A}_\pm}^2$ with the magnetic Schr\"odinger operator $\boldsymbol{p}_{\boldsymbol{A}_\pm}^2$
on an appropriate scale, as in~\cite{Sen}. To this end, we shall localize the operators to squares $S$ over which the $L_p$
average of $|B_\pm|$ is small compared to $\ell(S)^{-2/3}$ (the localization error in the kinetic energy). More precisely, we
divide $\oR^2$ into a grid of disjoint squares $\{S_j\}$ where each $S_j$ is a maximal dyadic square such that
\begin{align}
\left(\int_{12S(x,\ell)} |B_\pm(y)|^p\,\,dy\right)^{1/p} \leqslant \frac{\varepsilon}{[\ell(S)]^{(3p-4)/2p}}\,\,,
\label{smallE}
\end{align}
where $\ell(S_j)=\ell_j$ is the side length of $S_j$, and $12 \ell_j$ denotes the square which has the same center as
$S_j$ and side length $12\ell(S_j)$. It can be proved that $\ell_j \approx \ell_k$ if $4S_j \cap 4S_k \not= \varnothing$.
Then, using this property, one constructs a partition of unity for $\oR^2$: $\sum_j \phi_j \equiv 1$, with 
$\phi_j \in C_0(2S_j;\oR)$. This leads us to the following IMS-type localization formula, which in the non-relativistic case
says that for any $\Psi_\pm$ and $\boldsymbol{A}_\pm$
\begin{align}
\int_{\oR^2} \bigl[(\boldsymbol{p}-{\boldsymbol{A}_\pm})\Psi_\pm\bigr]^2\,\,dx
=\sum_j \int_{\oR^2} \bigl[(\boldsymbol{p}-{\boldsymbol{A}_\pm})(\phi_j \Psi_\pm)\bigr]^2\,\,dx
+\sum_j \int_{\oR^2} |\boldsymbol{\nabla}\phi_j|^2\,|\Psi_\pm|^2\,\,dx\,\,.
\label{IMSform}
\end{align}
In this case, the localization error $\sum_j |\boldsymbol{\nabla}\phi_j|^2$ is such that
$|\boldsymbol{\nabla}^\alpha \phi_j| \leqslant C_\alpha/\ell_j^{|\alpha|}$, being local and independent of ${\boldsymbol{A}_\pm}$.
With this localization formula, we will show that if $\varepsilon$ in (\ref{smallE}) is sufficiently small, then
\begin{align*}
\boldsymbol{p}_{\boldsymbol{A}_\pm}^2
\leqslant C\bigl\{\boldsymbol{P}_{\boldsymbol{A}_\pm}^2+{\boldsymbol{\varPhi}}\bigr\}\,\,,
\end{align*}
where ${\boldsymbol{\varPhi}}=\sum_j \phi_j^2/\ell_j^{3/2}$, such that $\boldsymbol{\varPhi}(x) \approx {\mathfrak b}_p(x)$
for $p > 4/3$. Here, as explained below, ${\mathfrak b}_p(x)$ is an effective (scalar) magnetic field defined to be the
$L_p$ average of $|B_\pm|$ over a suitable square centered at $x$ with a side length scaling like $|B_\pm|^{-2/3}$.

\subsection{Proof of Theorem \ref{Sobolev2}}
\label{Sec5c}
\hspace*{\parindent}
Compared to the articles by Sobolev~\cite{Sob} and Bugliaro {\em et al.}~\cite{BFFGS}, Shen~\cite{Sen} found a
simpler and more natural way to define the effective field. In $\oR^3$, Shen's idea was to replace ${\mathfrak b}(x)$
with ${\mathfrak b}_p(x)$, where the last is defined to be the $L_p$ average of $|B|$ over a suitable cube centered
at $x$ with a side length scaling like $|B|^{-1/2}$. Based on work by Shen~\cite{Sen}, our next goal is to provide lower
bounds for the sum of the negative eigenvalues of the two-dimensional massless Dirac operator with the magnetic
fields $B_\pm=eB^{\rm ext} \pm gb^{\rm ind}$, where $B^{\rm ext} > 0$ is a homogeneous field and $b^{\rm ind}$ a
non-homogeneous magnetic field, perturbed by the $K_0$-potential. With this in mind, since $|B_\pm|$ scales like
({\em length})$^{-3/2}$ in $\oR^2$, a simple dimension counting shows that ${\mathfrak b}_p(x)$ must be defined to
be the $L_p$ average of $|B_\pm|$ over a suitable square centered at $x$ with a side length scaling like
$|B_\pm|^{-2/3}$. 

\begin{remark}
At this point, a comment on the power $-2/3$ in $|B_\pm|^{-2/3}$ is in order. Since we are assuming that the field
$B^{\rm ext}$ points perpendicularly to the graphene sheet plane, that is, it points along the $x_3$-axis (which is always true
for two dimensions), we should have $\boldsymbol{B}^{\rm ext}=(0,0,B^{\rm ext})$. This implies that dimensionally $|B^{\rm ext}|$
scales like ({\em length})$^{-2}$, while $|b^{\rm ind}|$ scales like ({\em length})$^{-3/2}$, once $b^{\rm ind}$ is the induced
field within the bulk of the system. This would be so if the charge carriers in graphene were described by massless
Dirac fermions constrained to move on a two-dimensional (2D) manifold embedded in three-dimensional (3D) space.
This apparent discrepancy in the scaling is remedied by remembering that graphene is a {\em strictly two-dimensional
material}~\cite{Novo}, and it is the interaction between the quasi-particles that should determine the scaling of the field
$|B^{\rm ext}|$. In other words, the quasi-particles are no more able to perceive the third dimension than {\em Flatland's Square}.
Because of this, $|B^{\rm ext}|$ should scale like ({\em length})$^{-3/2}$. We call this the {\bf physical dimension},
that is, the dimension determined by the interaction between the quasi-particles.
\label{Wado}
\end{remark}

Taking Remark \ref{Wado} and Shen's work into account, a basic length scale can be defined as
\begin{align}
\ell_p(x)=\sup\left\{\ell > 0 \mid \ell^{3/2}
\left(\frac{1}{\ell^2}\int_{S(x,\ell)} |B_\pm(y)|^p\,\,dy\right)^{1/p} \leqslant 1\right\}\,\,,
\label{Shen1}
\end{align}
where $S(x,\ell)$ denotes the a square in $\oR^2$ centered at $x$ with side length $\ell$. Note that the Eq.(\ref{Shen1})
implies that
\begin{align*}
0 \leqslant \left(\int_{S(x,\ell)} |B_\pm(y)|^p\,\,dy\right)^{1/p} \leqslant \frac{1}{\ell^{(3p-4)/2p}}\,\,.
\end{align*}
Assuming that $p > 4/3$ and taking the limit $\ell \to \infty$ we found that $0 \leqslant \|B_\pm\|_p \leqslant 0$,
implying that $B_\pm$ should be identically zero. Hence, from now on we will assume that $B_\pm \not\equiv 0$
and $B_\pm \in L_p^{\rm loc}(\oR^2)$ for some $p > 4/3$. With this we have $0 < \ell_p(x) < \infty$ for any
$x \in \oR^2$. Thus, according to Shen~\cite{Sen}, our effective field is given by
\begin{align}
{\mathfrak b}_p(x)=\frac{1}{(\ell_p(x))^{3/2}}
=\left(\frac{1}{\ell_p^2(x)}\int_{S(x,\ell)} |B_\pm(y)|^p\,\,dy\right)^{1/p}\,\,.
\label{Shen1A}
\end{align}

\begin{proposition}
Given that $B_\pm=eB^{\rm ext} \pm gb^{\rm ind}$, suppose that $|B^{\rm ext}| \in L_\infty(\oR^2)$ and
$|b^{\rm ind}| \in L_p(\oR^2)$ for some $p > 4/3$. Let $\ell=\ell_p(x)$; then
\begin{align}
{\mathfrak b}_p(x)=\frac{1}{(\ell_p(x))^{3/2}} \leqslant C \Bigl(\|eB^{\rm ext}\|_\infty+\|gb^{\rm ind}\|_p^{3p/(3p-4)}\Bigr)\,\,.
\label{Shen2}
\end{align}
\label{Shen2A}
\end{proposition}

\begin{proof}
We start by defining the function $h=h(\ell)$, with $\ell > 0$, given by
\begin{align*}
h(\ell)=\ell^{3/2} \left(\frac{1}{\ell^2}\int_{S(x,\ell)} |B_\pm(y)|^p\,\,dy\right)^{1/p}\,\,.
\end{align*}
Note that $h$ is continuous if $p > 4/3$. Hence, if $\ell > 0$, it follows that $h(\ell) < 1$. Then, for $\varepsilon > 0$
arbitrarily small, it is also true that $h(\ell+\varepsilon) < 1$. Thus, if $\ell_p(x)$ is the $\sup$ of such $\ell > 0$ for which
$h(\ell) \leqslant 1$, by Eq.(\ref{Shen1}) we must have $h(\ell_p(x))=1$, {\em i.e.}, we have
\begin{align*}
1&=[\ell_p(x)]^{3/2} \left(\frac{1}{\ell_p^2(x)}\int_{S(x,\ell_p(x))} |B_\pm(y)|^p\,\,dy\right)^{1/p} \\[3mm]
&\leqslant [\ell_p(x)]^{(3p-4)/2p} \left\{\left(\int_{S(x,\ell_p(x))} |eB^{\rm ext}(y)|^p\,\,dy\right)^{1/p}
+\left(\int_{S(x,\ell_p(x))} |gb^{\rm ind}(y)|^p\,\,dy\right)^{1/p}\right\} \\[3mm]
&\leqslant [\ell_p(x)]^{(3p-4)/2p} \Bigl([\ell(x)]^{2/p} \|eB^{\rm ext}\|_\infty+\|gb^{\rm ind}\|_p \Bigr) \\[3mm]
&=[\ell_p(x)]^{3/2} \Bigl(\|eB^{\rm ext}\|_\infty+C'\|gb^{\rm ind}\|_p^{3p/(3p-4)}\Bigr)\,\,,
\end{align*}
where we used Minkowski in the first inequality and define
\begin{align*}
C'=\left(\frac{1}{\ell_p(x) \|gb^{\rm ind}\|_p^{2p/(3p-4)}}\right)^{2/p}\,\,.
\end{align*}
Finally, Eq.(\ref{Shen2}) then follows by taking $C=\max\,\{1,C'\}$.
\end{proof}

Following the script of Shen~\cite{Sen}, next, we shall sketch the construction of the partition of unity for $\oR^2$
associated with the function $\ell_p(x)$. First, we define ${\mathscr A}$ the set of all dyadic squares in $\oR^2$
such that
\begin{align}
\left(\int_{12S(x,\ell)} |B_\pm(y)|^p\,\,dy\right)^{1/p} \leqslant \frac{\varepsilon}{[\ell(S)]^{(3p-4)/2p}}\,\,,
\label{ShenEp}
\end{align}
where $\varepsilon \in (0,1)$ is a constant to be determined later. Here, $\ell(S)$ denotes the side length of $S$. 
It is said that $S$ is a {\em maximal element} of ${\mathscr A}$ if $S \in {\mathscr A}$ and $S$ is not properly
contained in any other square in ${\mathscr A}$. Let ${\mathscr B}$ denote the set of all maximal elements of
${\mathscr A}$. By definition, the interiors of the squares in ${\mathscr B}$ are disjoint.
Let ${\mathscr B}=\{S_j\}_{j=1}^\infty$; then (see~\cite[Lemma 3.1]{Sen})
\begin{align*}
\oR^2=\bigcup_j S_j\,\,.
\end{align*}
In the sequel, using the same argument as in~\cite[Lemma 3.2]{Sen} one shows that
\begin{align}
\frac{1}{2} \ell(S_k) \leqslant \ell(S_j) \leqslant 2 \ell(S_k)\,\,, 
\quad \text{if} \quad 4S_j \cap 4S_k \not= \varnothing\,\,.
\label{Shen3} 
\end{align}

It follows from (\ref{Shen3}) the

\begin{lemma}[Shen~\cite{Sen}, p.318]
There exists a sequence of functions $(\phi_j)_{j \in \oN}$ such that
\begin{enumerate}
\item $\phi_j \in C_0^\infty(2S_j;\oR)$ and $0 \leqslant \phi_j \leqslant 1$;

\item $|\boldsymbol{\nabla}^\alpha\phi_j| \leqslant C_\alpha/\ell_j^{|\alpha|}$, where $\ell_j=\ell(S_j)$;

\item $\sum_j \phi_j^2 \equiv 1$ in $\oR^2$.
\end{enumerate}
\label{LemmaShen}
\end{lemma}

Recall that $\ell_j=\ell(S_j)$ and $S_j$ is a maximal square. Define
\begin{align*}
\boldsymbol{\varPhi}(x)=\sum_j \frac{1}{\ell_j^{3/2}}\phi_j^2(x)\,\,.
\end{align*}

The next result is a 2D version of Shen's Theorems 3.1 and 5.1~\cite{Sen}.

\begin{theorem}
There exist constants $C > 0$ and $\varepsilon_0 > 0$ such that, if $0< \varepsilon < \varepsilon_0$, 
we have
\begin{align*}
\boldsymbol{p}_{\boldsymbol{A}_\pm}^2
\leqslant C\bigl\{\boldsymbol{P}_{\boldsymbol{A}_\pm}^2+{\boldsymbol{\varPhi}}\bigr\}\,\,.
\end{align*}
\label{ShenTheo}
\end{theorem}

\begin{proof}
As in the proof of Theorems 3.1 and 5.1 in~\cite{Sen}, we start by showing that,
if $\varepsilon$ in (\ref{ShenEp}) is small, then
\begin{align}
\int_{\oR^2} ||B_\pm|^{1/2}\phi_j\Psi_\pm|^2\,\,dx
\leqslant \int_{\oR^2} |\boldsymbol{P}_{\boldsymbol{A}_\pm}(\phi_j\Psi_\pm)|^2\,\,dx\,\,,
\label{ShenTheoP1}
\end{align}
for any $\Psi_\pm \in C_0^\infty(\oR^2;\oC)$. To show (\ref{ShenTheoP1}), we shall use the generalized
Lad\'yzhenskaya inequality for 2D~\cite{Olga} proved by Constantin-Seregin~\cite{ConSer}
\begin{align}
\|\Psi\|_{2q}^4 \leqslant \frac{q}{\sqrt{2}}  \|\Psi\|_2^{2} \|\boldsymbol{\nabla}\Psi\|_q^{2}\,\,,
\quad \text{$\forall\,\Psi \in C_0^\infty(\oR^2)$ and $q \geqslant 2$}\,\,,
\label{ShenTheoP2}
\end{align}
This inequality is a special case of the Gagliardo-Nirenberg interpolation inequality.
Let $1/p+1/q=1$. It follows that for the left side of (\ref{ShenTheoP1}), by H\"older inequality, we have
\begin{align*}
\int_{\oR^2} ||B_\pm|^{1/2}\phi_j\Psi_\pm|^2\,\,dx
&=\int_{2S_j} |B_\pm| |\phi_j\Psi_\pm|^2\,\,dx \\[3mm]
&\leqslant \left(\int_{2S_j} |B_\pm|^2\,\,dx\right)^{1/2} \left(\int_{2S_j} |\phi_j\Psi_\pm|^{4}\,\,dx\right)^{1/2}\,\,.
\end{align*}
By Eq.(\ref{ShenEp}) the first term on the right side of the above equation is estimated by
\begin{align*}
\left(\int_{2S_j} |B_\pm|^2\,\,dx\right)^{1/2}
\leqslant \frac{\varepsilon}{\sqrt{2}}\,\,,
\end{align*}
while for the second term it follows that
\begin{align*}
\left(\int_{2S_j} |\phi_j\Psi_\pm|^{4}\,\,dx\right)^{1/2}
=\left[\left(\int_{2S_j} |\phi_j\Psi_\pm|^{4}\,\,dx\right)^{1/4}\right]^2
&\leqslant \left[\left(\int_{2S_j} |\phi_j\Psi_\pm|^{4}\,\,dx\right)^{1/4}\right]^4 \\[3mm]
&\leqslant \frac{2}{\sqrt{2}} \left(\int_{2S_j} |\phi_j\Psi_\pm|^{2}\,\,dx\right) 
\left(\int_{2S_j} |\boldsymbol{\nabla}\phi_j\Psi_\pm|^{2}\,\,dx\right) \\[3mm]
&\leqslant C \frac{2}{\sqrt{2}} \int_{2S_j} |\boldsymbol{\nabla}\phi_j\Psi_\pm|^{2}\,\,dx \\[3mm]
&\leqslant C \frac{2}{\sqrt{2}} \int_{\oR^2} |\boldsymbol{\nabla}\phi_j\Psi_\pm|^{2}\,\,dx \\[3mm]
&=C \frac{2}{\sqrt{2}} \int_{\oR^2} |\boldsymbol{\nabla}|\phi_j\Psi_\pm||^{2}\,\,dx\,\,,
\end{align*}
where we used (\ref{ShenTheoP2}). In this way, we get
\begin{align*}
\int_{\oR^2} ||B_\pm|^{1/2}\phi_j\Psi_\pm|^2\,\,dx
&\leqslant C \varepsilon \int_{\oR^2} |\boldsymbol{\nabla}|\phi_j\Psi_\pm||^{2}\,\,dx \\[3mm]
&\leqslant C \varepsilon \int_{\oR^2} |\boldsymbol{p}_{\boldsymbol{A}_\pm}(\phi_j\Psi_\pm)|^{2}\,\,dx \\[3mm]
&\leqslant C \varepsilon \int_{\oR^2} |\boldsymbol{P}_{\boldsymbol{A}_\pm}(\phi_j\Psi_\pm)|^{2}\,\,dx
+C \varepsilon \int_{\oR^2} ||B_\pm|^{1/2}\phi_j\Psi_\pm|^2\,\,dx\,\,,
\end{align*}
where we used the diamagnetic inequality~\cite{LiLo,LiSe}
$|\boldsymbol{\nabla}|\phi_j\Psi_\pm|| \leqslant |\boldsymbol{p}_{\boldsymbol{A}_\pm} (\phi_j\Psi_\pm)|$ in the
second inequality. Eq.(\ref{ShenTheoP1}) then follows by choosing $\varepsilon$ small so that $C\varepsilon < 1/2$.

Next we use Lemma \ref{LemmaShen}, (\ref{IMSform}) and (\ref{ShenTheoP1}) to obtain
\begin{align*}
\int_{\oR^2} |\boldsymbol{p}_{\boldsymbol{A}_\pm} \Psi_\pm|^{2}\,\,dx
&=\sum_j \int_{\oR^2} |\boldsymbol{p}_{\boldsymbol{A}_\pm}(\phi_j\Psi_\pm)|^{2}\,\,dx
+\sum_j \int_{\oR^2} |\boldsymbol{\nabla}\phi_j|^2\,|\Psi_\pm|^2\,\,dx \\[3mm]
&\leqslant \sum_j \int_{\oR^2} |\boldsymbol{P}_{\boldsymbol{A}_\pm}(\phi_j\Psi_\pm)|^{2}\,\,dx
+\sum_j \int_{\oR^2} ||B_\pm|^{1/2}\phi_j\Psi_\pm|^2\,\,dx \\[3mm]
&\qquad+\sum_j \int_{\oR^2} |\boldsymbol{\nabla}\phi_j|^2\,|\Psi_\pm|^2\,\,dx \\[3mm]
&\leqslant C \sum_j \int_{\oR^2} |\boldsymbol{P}_{\boldsymbol{A}_\pm}(\phi_j\Psi_\pm)|^{2}\,\,dx
+\sum_j \int_{\oR^2} |\boldsymbol{\nabla}\phi_j|^2\,|\Psi_\pm|^2\,\,dx \\[3mm]
&\leqslant C \sum_j \int_{\oR^2} |\boldsymbol{P}_{\boldsymbol{A}_\pm} \Psi_\pm|^{2}\,\,dx
+C\sum_j \frac{1}{\ell_j^2} \int_{2S_j} |\Psi_\pm|^2\,\,dx \\[3mm]
&\leqslant C \sum_j \int_{\oR^2} |\boldsymbol{P}_{\boldsymbol{A}_\pm} \Psi_\pm|^{2}\,\,dx
+C\sum_j \frac{1}{\ell_j^{3/2}} \int_{2S_j} |\Psi_\pm|^2\,\,dx\,\,.
\end{align*}

The theorem follows if we have
\begin{align}
\sum_j \frac{1}{\ell_j^{3/2}} \chi_{2S_j} \leqslant C \boldsymbol{\varPhi}\,\,,
\quad \text{where} \quad
\chi _{2S_j}(x)=
\begin{cases}
1 & {\rm if} \quad x \in 2S_j \\[3mm]
0 & {\rm if} \quad x \notin 2S_j
\end{cases}\,\,.
\label{ShenTheoP3}
\end{align}
Exactly as in the proof of Theorem 5.1 in~\cite{Sen}, notice that (\ref{Shen3}) implies that if $x \in S_k$,
the l.h.s. of (\ref{ShenTheoP3}) is bounded by
\begin{align*}
\sum_{2S_j \cap S_k \not= \varnothing} \frac{1}{\ell_j^{2}}
\leqslant C \frac{1}{\ell_j^{2}} \cdot \# \Bigl\{S_j \mid 2S_j \cap S_k \not= \varnothing\Bigr\}
\leqslant C \frac{1}{\ell_k^{2}} \leqslant C \frac{1}{\ell_k^{3/2}} \leqslant C \boldsymbol{\varPhi}\,\,.
\end{align*}
The proof is now complete.
\end{proof}

Finally we are in a position to give the

\begin{proof}[Proof of Theorem \ref{Sobolev2}]
Our proof proceeds along the same lines as Sobolev~\cite{Sob} and Shen~\cite{Sen}, considering the case $\gamma=1$
only. For $\lambda > 0$, we denote by $N(\lambda;|\boldsymbol{P}_{\boldsymbol{A}_\pm}|+V_s)$ the number of eigenvalues
(counting multiplicity) of $|\boldsymbol{P}_{\boldsymbol{A}_\pm}|+V_s$ smaller than $-\lambda$. Noting that for any
potential $V$ and $\lambda > 0$
\begin{align*}
V+\lambda \geqslant -V_-+\lambda \geqslant -(V+\lambda/2)_-+\lambda/2\,\,,
\end{align*}
the variational principle says that the operator inequalities
\begin{align*}
f(\boldsymbol{p}_{\boldsymbol{A}_\pm})+V+\lambda
\geqslant
f(\boldsymbol{p}_{\boldsymbol{A}_\pm})-V_-+\lambda
\geqslant
f(\boldsymbol{p}_{\boldsymbol{A}_\pm})-(V+\lambda/2)_-
+\lambda/2\,\,,
\end{align*}
imply that
\begin{align*}
N(\lambda;f(\boldsymbol{p}_{\boldsymbol{A}_\pm})+V)
\leqslant
N(\lambda/2;f(\boldsymbol{p}_{\boldsymbol{A}_\pm})-(V+\lambda/2)_-)\,\,,
\end{align*}
and hence by the BKS inequalities (\ref{BKSinq}), with $\alpha=1/2$, we have
\begin{align*}
{\rm Tr}(|\boldsymbol{P}_{\boldsymbol{A}_\pm}|+V_s)
={\rm Tr}(|\boldsymbol{P}_{\boldsymbol{A}_\pm}|-V_-) 
&\leqslant {\rm Tr}(\boldsymbol{P}_{\boldsymbol{A}_\pm}^2-V_-^2)^{1/2} \\[3mm]
&=\frac{1}{2} \int_0^\infty \lambda^{-1/2}
\,N(\lambda;(\boldsymbol{P}_{\boldsymbol{A}_\pm}^2-V_-^2)\,\,d\lambda \\[3mm]
&\leqslant \frac{1}{2} \int_0^\infty \lambda^{-1/2}
\,N(\lambda/2;\boldsymbol{P}_{\boldsymbol{A}_\pm}^2-(V^2+\lambda/2)_-)\,\,d\lambda \\[3mm]
&=\frac{1}{\sqrt{2}} \int_0^\infty \mu^{-1/2}
\,N(\mu;\boldsymbol{P}_{\boldsymbol{A}_\pm}^2-(V^2+\mu)_-)\,\,d\mu\,\,.
\end{align*}

Now, we use the running-energy-scale method~\cite{LLS}, in a slightly different way as suggested
by Bley-Fournais~\cite[Theorem 2.9]{GoFo}, and Theorem \ref{ShenTheo}, to write the last integral above as
\begin{align*}
\frac{1}{\sqrt{2}} \int_0^\infty \mu^{-1/2}
\,N(\mu;\boldsymbol{P}_{\boldsymbol{A}_\pm}^2-(V^2+\mu)_-)\,\,d\mu
&\leqslant \frac{1}{\sqrt{2}} \int_0^\infty \mu^{-1/2}
\,N(\mu;\zeta\boldsymbol{P}_{\boldsymbol{A}_\pm}^2-(V^2+\mu)_-)\,\,d\mu \\[3mm]
&\leqslant \frac{1}{\sqrt{2}} \int_0^\infty \mu^{-1/2}
\,N(\mu;\zeta C^{-1}\boldsymbol{p}_{\boldsymbol{A}_\pm}^2-\zeta \boldsymbol{\varPhi}-(V^2+\mu)_-)\,\,d\mu \\[3mm]
&= \frac{1}{\sqrt{2}} \int_0^\infty \mu^{-1/2}
\,N(\mu;\boldsymbol{p}_{\boldsymbol{A}_\pm}^2-C\boldsymbol{\varPhi}-(V^2/\zeta+\mu/\zeta)_-)\,\,d\mu\,\,,
\end{align*}
where $\zeta \in (0,1)$ is a parameter which is to be optimized later. By treating the term $\boldsymbol{\varPhi}$
as a potential, in the sequel we shall use the bound (for the magnetic momentum $\boldsymbol{p}_{\boldsymbol{A}_\pm}$)
\begin{align}
N_\alpha(V) \leqslant C_n \int_{\oR^n} (V(x)-\alpha)_-^{n/2}\,\,dx\,\,,
\label{Daub}
\end{align}
for some constant $C_n$ (see Daubechies~\cite[Remark 5, p.517]{Dau}). In fact, since
$(\boldsymbol{p}_{\boldsymbol{A}_\pm}+\lambda)^{-1} \leqslant (\boldsymbol{p}+\lambda)^{-1}$ we immediately have
${\rm Tr}(f(K_\lambda(V;\boldsymbol{A}_\pm))) \leqslant {\rm Tr}(f(K_\lambda(V;0))$, where $K_\lambda$ is the
Birman-Schwinger kernel, and this implies that $N_\alpha(V;\boldsymbol{A}_\pm) \leqslant C_n \int_{\oR^n} (V(x)-\alpha)_-^{n/2}$,
with the same constant as in (\ref{Daub}) (however, it is important to point out that it is not true in general that
$N_\alpha(V;\boldsymbol{A}_\pm) \leqslant N_\alpha(V;0)$; see Example 2 following Theorem 2.14 in~\cite{AHS}).

Returning to our proof, using the bound (\ref{Daub}) it follows that
\begin{align}
{\rm Tr}(|\boldsymbol{P}_{\boldsymbol{A}_\pm}|+V_s) 
\leqslant \frac{C_2}{\sqrt{2}} \int_0^\infty \mu^{-1/2}
\left(\int_{\oR^2} \,(-C\boldsymbol{\varPhi}+V^2/\zeta+\mu/\zeta)_-\,\,dx\right)\,\,d\mu\,\,.
\label{Bley}
\end{align}

Note that, with the change of variable $\eta=\mu/\zeta$ in (\ref{Bley}), from Remark \ref{Rema}, for any fixed $x$,
only $\eta$'s with $0 < \eta < V_-^2/\zeta+C\boldsymbol{\varPhi}$ contribute to the integral in the variable $\eta$.
Therefore, changing the order of integration, we have
\begin{align*}
{\rm Tr}(|\boldsymbol{P}_{\boldsymbol{A}_\pm}|+V_s) 
\leqslant \frac{C_2}{\sqrt{2}} \int_{\oR^2}
\left(\zeta^{1/2} \int_0^{V_-^2/\zeta+C\boldsymbol{\varPhi}} \,\eta^{-1/2}(-C\boldsymbol{\varPhi}+V^2/\zeta+\eta)_-\,\,d\eta\right)\,\,dx\,\,.
\end{align*}
with the integral in the variable $\eta$ being calculated through the following formula:
\begin{align*}
\int_0^{u} dx\,\,x^{\nu-1} (u-x)^{\mu-1}=u^{\mu+\nu-1} \frac{\Gamma(\mu)\Gamma(\nu)}{\Gamma(\mu+\nu)}\,\,,
\quad \text{[${\rm Re}\,\mu > 0, {\rm Re}\,\nu > 0$]}\,\,.
\end{align*}
This gives us
\begin{align*}
{\rm Tr}(|\boldsymbol{P}_{\boldsymbol{A}_\pm}|+V_s) 
&\leqslant \frac{C_2}{\sqrt{2}} \frac{\Gamma(2)\Gamma(1/2)}{\Gamma{(5/2)}} \zeta^{1/2}
\int_{\oR^2} (V_-^2/\zeta+C\boldsymbol{\varPhi})^{3/2}\,\,dx \\[3mm]
&\leqslant \frac{2^{3/2}C_2}{3}
\left(\zeta^{-1} \int_{\oR^2} V_-^3\,\,dx+\zeta^{1/2} C^{3/2} \int_{\oR^2} \boldsymbol{\varPhi}^{3/2}\,\,dx\right) \\[3mm]
&\equiv \frac{2^{3/2}C_2}{3} \left(\omega^{-2}M+\omega N\right)\,\,,
\end{align*}
where $C_2$ is the constant in the bound (\ref{Daub}). We follow on optimizing this with respect to $\omega$; in fact,
we see that the optimal $\omega$ is
\begin{align*}
\min\,\left\{1,\left(\frac{2M}{N}\right)^{1/3}\right\}\,\,,
\end{align*}
and as a consequence
\begin{align*}
\omega^{-2}M+\omega N &\leqslant
\begin{cases}
2^{-2/3}3M^{1/3}N^{2/3}\,\,, & \text{if} \quad 2M \leqslant N \\[3mm]
3M\,\,, & \text{if} \quad 2M > N
\end{cases} \\[3mm]
&\leqslant 2^{-2/3}3M^{1/3}N^{2/3}+3M\,\,.
\end{align*}
So we have
\begin{align}
{\rm Tr}(|\boldsymbol{P}_{\boldsymbol{A}_\pm}|+V_s) \leqslant
C_{1} \int_{\oR^2} V_-^3(x)\,\,dx
+C'_{2} \left(\int_{\oR^2} \boldsymbol{\varPhi}^{3/2}(x)\,\,dx\right)^{2/3}
\left(\int_{\oR^2} V_-^3(x)\,\,dx\right)^{1/3}\,\,,
\label{MODE}
\end{align}
where $C_1 \overset{\rm def.}{=} 2^{3/2} C_2$ and $C'_2 \overset{\rm def.}{=} 2^{1/2} C_2 C^{3/2}$.

As a penultimate step, taking into account that $V_-(x)=\gamma K_0(\beta |x|)$ and that $K_0(\beta |x|)$
is a positive function, we will use a version of the reverse H\"older inequality~\cite{Bene} (in the literature
there are many versions of this inequality) which establishes the following: consider H\"older conjugate 
exponents $p$ and $q$, with $p^{-1}+q^{-1}=1$, and assume that $p \in (0,1)$. Let $(\cal{M},\mu)$ be a
measure space, with $\mu(M) > 0$; then for all measurable real- or complex-valued functions $f$ and $g$
on $\cal{M}$, such that $g(x) > 0$ for $\mu$-almost all $x \in \cal{M}$,
\begin{align*}
\left(\int _{\cal{M}} |f(x)|^{\frac{1}{p}}\,\,d\mu(x)\right)^{p} \left(\int _{\cal{M}} |g(x)|^{\frac{1}{q}}\,\,d\mu(x)\right)^{q}
\leqslant \|fg\|_1\,\,.
\end{align*}

Armed with this result, notice that for $p=2/3$, the second term in (\ref{MODE})
$\leqslant C'_{2} \int_{\oR^2} \boldsymbol{\varPhi}(x) V_-(x)\,\,dx$. Then, by an argument similar to that in~\cite[Lemma 4.5]{Sen},
we have that $\boldsymbol{\varPhi}(x) \approx {\mathfrak b}_p(x)$ for $p > 4/3$, and from Proposition \ref{Shen2A},
we get
\begin{align}
{\rm Tr}(|\boldsymbol{P}_{\boldsymbol{A}_\pm}|+V_s) 
&\leqslant C_{1} \int_{\oR^2} V_-^3(x)\,\,dx
+C'_{2} \int_{\oR^2} {\mathfrak b}_p(x) V_-(x)\,\,dx \nonumber \\[3mm]
&\leqslant C_{1} \int_{\oR^2} V_-^3(x)\,\,dx
+C''_{2} \Bigl(\|eB^{\rm ext}\|_\infty+\|gb^{\rm ind}\|_p^{3p/(3p-4)}\Bigr)
\int_{\oR^2} V_-(x)\,\,dx\,\,.
\label{ErdosShen}
\end{align}

Finally, using polar coordinates, with $V_-(x)=\gamma_s K_0(\beta |x|)$, we obtain:  $\int_{\oR^2} V_-(x)=2\pi \gamma_s \beta^{-2}$
and $\int_{\oR^2} V_-^3(x)=2\pi \gamma_s^3 \beta^{-2}(1/6\psi'(1/3)-\pi^2/9)$. These integrals are obtained, respectively, 
in~\cite[Formula {\bf 6.561}, 16., p.676]{Grad} and in~\cite[Formula {\bf 4.10.2.2.}, p.216]{Yuri}. This concludes the proof.
\end{proof}

\begin{remark}
By adopting inequality (\ref{BKSinq}) as one of the strategies of proof of Theorem \ref{Sobolev2}, we saw that it
has been effectively possible to replace $|\boldsymbol{P}_{\boldsymbol{A}_\pm}|$
by $\boldsymbol{P}_{\boldsymbol{A}_\pm}^2=\boldsymbol{p}_{\boldsymbol{A}_\pm}^2-B_\pm$, but this has a direct
impact on the model describing pristine graphene as proposed in Refs.~\cite{{WOE,Em}}. As we are assuming that
$B^{\rm ext} > 0$ and $b^{\rm ind}$ is a positive function, the total field $B_\pm=eB^{\rm ext} \pm gb^{\rm ind}$ can
be positive or negative depending on the sign of the charges $e$ and $g$ of the quasi-particles (see Table \ref{table1}).
In this case, for the $s$-wave state, the upper bound (\ref{ErdosShen}) for the bound states of
(electron-polaron)--(electron-polaron) does not contain the term depending on $B^{\rm ext}$ and $b^{\rm ind}$.
This is because the sum of negative eigenvalues of the operator $\boldsymbol{p}_{\boldsymbol{A}_\pm}^2+B_\pm+V_s$
can be estimated by sum of negative eigenvalues of the operator $\boldsymbol{p}_{\boldsymbol{A}_\pm}^2+V_s$. The
situation is reversed if we adopt $B^{\rm ext} < 0$; in this case, for the $s$-wave state, the upper bound (\ref{ErdosShen})
for the bound states of (hole-polaron)--(hole-polaron) does not contain the term depending on $B^{\rm ext}$ and $b^{\rm ind}$.
\label{RemaBS}
\end{remark}

An important and non-trivial consequence of the Theorem \ref{Sobolev2} are the bounds on the spectrum of
$|\boldsymbol{P}_{\boldsymbol{A}_\pm}|+V_s$, namely
\begin{align*}
|\boldsymbol{P}_{\boldsymbol{A}_\pm}|-\gamma_s K_0(\beta|x|)
&\geqslant -\sum_k |\lambda_k^\pm(|\boldsymbol{P}_{\boldsymbol{A}_\pm}|+V_s)| \\[3mm]
&\geqslant -C_{1} \gamma_s^3\,\beta^{-2}-C''_{2} \gamma_s\,\beta^{-2}\Bigl(\|eB^{\rm ext}\|_\infty
+\|gb^{\rm ind}\|_p^{3p/(3p-4)}\Bigr)\,\,.
\end{align*}
The first inequality just expresses the fact that, because of the Pauli principle for fermions, the
lowest possible energy (left hand side of above equation) is obtained when the quasi-particles assume
the states of $s$-wave, as mentioned in the Introduction. The second inequality concerns the trace over
the negative spectrum of the operator $|\boldsymbol{P}_{\boldsymbol{A}_\pm}|+V_s$ analyzed in
Theorem \ref{Sobolev2}. Therefore, the bound state formed by a quasi-particle tied to a field of another
quasi-particle, and subject to relativistic effects, is stable if the inequality
\begin{align}
\Biggl\langle \Psi,
\Biggl[|\boldsymbol{P}_{\boldsymbol{A}_\pm}|-\gamma_s K_0(\beta|x|)+C_{1} \gamma_s^3\,\beta^{-2}
+C''_{2} \gamma_s\,\beta^{-2}\Bigl(\|eB^{\rm ext}\|_\infty+\|gb^{\rm ind}\|_p^{3p/(3p-4)}\Bigr)\Biggr]\Psi \Biggr\rangle
\geqslant 0\,\,,
\label{StableCond}
\end{align}
is satisfied. In Eq.(\ref{StableCond}) we can allow the constant in front of $K_0(\beta|x|$) to go all the way
to the critical constant 
\begin{align*}
\gamma_{\rm crit}=\left(\frac{2 (2\pi)^{3/2} e^{1/2}}{[\Gamma(1/4)]^4}\right) \beta
\approx 0,3 \beta\,\,,
\quad \text{(in units so that $\hbar=v_F=1$)}\,\,,
\end{align*}
obtained in~\cite{MOD2} through the relativistic Hardy inequality, {\em i.e.}, we can take $\gamma_s=\frac{1}{2\pi}(g^2-e^2)$
and $\beta=\frac{3,33}{2\pi}\,(g^2-e^2)$ in Eq.(\ref{StableCond}) in order to obtain
\begin{align}
\Biggl\langle \Psi,
\Biggl[|\boldsymbol{P}_{\boldsymbol{A}_\pm}|&-\frac{1}{2\pi}(g^2-e^2) K_0(\beta|x|)
+\frac{3,33}{2\pi}\,C_{1}(g^2-e^2) \nonumber \\[3mm]
&+\frac{0,6\pi\,C''_{2}}{(g^2-e^2)} \Bigl(\|eB^{\rm ext}\|_\infty+\|gb^{\rm ind}\|_p^{3p/(3p-4)}\Bigr)\Biggr]\Psi \Biggr\rangle
\geqslant 0\,\,.
\label{StableCond1}
\end{align} 
Therefore, we see that the state of a quasi-particle tied to a field of another quasi-particle, and subject to
relativistic effects, is stable when the constant $\gamma_s$ is allowed to go all the way to the critical constant
$\gamma_{\rm crit}$, as long as the field energy is added with $|g| > |e|$ and with sufficiently large positive
constants $C_1$ and $C_2$ in front; in other words, we have the

\begin{proposition}[Stability of an artificial atom in magnetic fields]
Assume that $\gamma_s=\gamma_{\rm crit}$. Then, for $|g| > |e|$ and sufficiently large positive
constants $C_1$ and $C_2$, it follows that
\begin{align*}
|\boldsymbol{P}_{\boldsymbol{A}_\pm}|-\frac{1}{2\pi}(g^2-e^2) K_0(\beta|x|)
+\frac{3,33}{2\pi}\,C_{1}(g^2-e^2)
+\frac{0,6\pi\,C_{2}}{(g^2-e^2)} \Bigl(\|eB^{\rm ext}\|_\infty+\|gb^{\rm ind}\|_p^{3p/(3p-4)}\Bigr)
\geqslant 0\,\,.
\end{align*}
\end{proposition}

The non-triviality of the above proposition lies in the fact that the operator
$|\boldsymbol{P}_{\boldsymbol{A}_\pm}|-\gamma_s K_0(\beta|x|)$ is not positive for any magnetic field.
It is the existence of zero modes of the operator $\boldsymbol{P}_{\boldsymbol{A}_\pm}$ which makes this
possible: a zero mode of $\boldsymbol{P}_{\boldsymbol{A}_\pm}$ is an eigenvector $\Psi$ corresponding to an
eigenvalue at 0, thus $\boldsymbol{P}_{\boldsymbol{A}_\pm}\Psi=0$ (despite the existence of zero modes of 
Weyl-Dirac operators be a rare phenomenon~\cite{BaEv}).

\begin{remark}
Another interesting problem to ask is whether a constant greater than a critical constant could compromise
the above stability. As first observed in an accompanying paper~\cite{MOD2}, analyzing the two-dimensional
Brown-Ravenhall operator with an attractive potential of the Bessel-Macdonald type, the kinetic term will always
dominate $V_-(x)=\gamma_s K_0(\beta|x|)$, regardless of the value of the coupling constant $\gamma_s$.
This result is a characteristic of the $K_0$-potential. The reason for this behavior is that unlike to usual Weyl-Dirac
operator with coulombian potential, the Weyl-Dirac operator with $K_0$-potential is not homogeneous with respect
to scalings of $\oR^2$, {\em i.e.}, the kinetic energy does not have the same behavior under scaling as the
Bessel-Macdonald energy for large momenta. In light of this result, we can conclude that a complete implosion of
an artificial atom in magnetic fields will never occur for two quasi-particles interacting via an attractive potential of
the Bessel-Macdonald type.
\end{remark}

\section*{Acknowledgements}
\hspace*{\parindent}
Oswaldo M. Del Cima and Daniel H.T. Franco thank Eduardo N.D. de Ara\'ujo for stimulating discussions on the physics
of graphene. Daniel H.T. Franco thanks the private communications with R. Frank, who indicated the
reading of articles by L. Erd\H{o}s and J.P. Solovej (through which we became aware of the articles by A. Sobolev, L. Bugliaro
{\em et. al.} and Z. Shen), and with S. Fournais about his article in collaboration with G. Bley. Emmanuel A. Pereira 
was partially supported by the Conselho Nacional de Desenvolvimento Cient\'\i fico e Tecnol\'ogico (CNPq).

\section*{Author's contributions}
\hspace*{\parindent}
All authors contributed equally to this work. On behalf of all authors, the corresponding author states that
there is no conflict of interest. 

\section*{Data availability}
\hspace*{\parindent}
The data that support the findings of this study are available from the corresponding author upon reasonable request.



\end{document}